\documentclass[10pt,fullpage,letterpaper]{article}

\usepackage{algorithm}
\usepackage[noend]{algorithmic}
\usepackage{enumerate}
\usepackage{amsmath, amsfonts, amssymb}
\usepackage[table]{xcolor}
\usepackage{amsthm}
\usepackage{mathrsfs}
\usepackage{bm}
\usepackage{graphicx}
\usepackage{tikz}
\usepackage{float}
\usepackage{booktabs}
\usepackage{hyperref}
\usepackage{array}
\usepackage{dsfont}
\usepackage{nicefrac}
\usepackage{comment}
\usepackage{anyfontsize}
\usepackage{longtable}
\usepackage{natbib}
\usepackage[symbol]{footmisc}
\allowdisplaybreaks
\usepackage{mathtools}
\usepackage{tablefootnote}

\DeclareMathOperator*{\argmin}{argmin}
\DeclareMathOperator*{\argmax}{argmax}

\newcommand{\mc}[1]{\mathcal{#1}}
\newcommand{\Sp}[1]{\left(#1\right)}
\newcommand{\Mp}[1]{\left[#1\right]}
\newcommand{\Bp}[1]{\left\{#1\right\}}
\newcommand{\abs}[1]{\left|#1\right|}
\newcommand{\Norm}[1]{\left\|#1\right\|}

\newcommand{\inner}[1]{\left\langle#1\right\rangle}
\newcommand{\ov}{\overline{V}}
\newcommand{\uv}{\underline{V}}
\newcommand{\oq}{\overline{Q}}
\newcommand{\uq}{\underline{Q}}

\newcommand{\A}{\mathcal{A}}

\newcommand{\E}{\mathbb{E}}
\newcommand{\G}{\mathcal{G}}
\renewcommand{\P}{\mathbb{P}}
\renewcommand{\S}{\mathcal{S}}

\renewcommand{\a}{\bm{a}}
\renewcommand{\r}{\bm{r}}
\newcommand{\R}{\mathbb{R}}
\newcommand{\N}{\mathcal{N}}

\newcommand{\M}{\mathcal{M}}
\newcommand{\F}{\mathcal{F}}

\newcommand{\pik}{\pi^k}

\newcommand{\pikni}{\pi^k_{-i}}
\newcommand{\hatphk}{\widehat{\mathbb{P}}_h^k}
\newcommand{\shk}{s_h^k}
\newcommand{\shkf}{s_h^{k, f}}
\newcommand{\ahk}{\bm{a}_h^k}
\newcommand{\Nhkf}{N_h^{k, f}}
\newcommand{\tq}{\widetilde{Q}}
\newcommand{\tv}{\widetilde{V}}

\newcommand{\bhk}{b_h^k}

\newcommand{\tmco}{\widetilde{\mc{O}}}
\newcommand{\pv}{\mathrm{pv}}
\newcommand{\reward}{\mathrm{r}}
\newcommand{\Ahik}{A_{h, i}^k}

\newcommand{\td}{\tilde{d}}
\newcommand{\wtthe}{\widetilde{\theta}}
\newcommand{\RNum}[1]{\uppercase\expandafter{\romannumeral #1\relax}}

\newtheorem{theorem}{Theorem}
\newtheorem{lemma}{Lemma}

\theoremstyle{definition}
\newtheorem{definition}{Definition}

\theoremstyle{remark}
\newtheorem{example}{Example}
\newtheorem{remark}{Remark}

\usepackage{color}
\definecolor{ForestGreen}{rgb}{0.1333,0.5451,0.1333}
\definecolor{Gray}{gray}{0.85}
\hypersetup{colorlinks,
	linkcolor=ForestGreen,
	citecolor=ForestGreen,
	urlcolor=black,
	linktocpage,
	plainpages=false}

\newcommand{\revision}[1]{{#1}}

\title{Learning in Congestion Games with Bandit Feedback}
\author{Qiwen Cui\footnote{Equal contribution}\\ \url{qwcui@cs.washington.edu} \and Zhihan Xiong\footnotemark[\value{footnote}] \\ \url{zhihanx@cs.washington.edu} \and Maryam Fazel \\ \url{mfazel@uw.edu} \and Simon S. Du \\\url{ssdu@cs.washington.edu}}
\date{}

\begin{document}
\maketitle

\renewcommand*{\thefootnote}{\arabic{footnote}}

\begin{abstract}
In this paper, we investigate Nash-regret minimization in congestion games, a class of games with benign theoretical structure and broad real-world applications.
We first propose a centralized algorithm based on the optimism in the face of uncertainty principle for congestion games with (semi-)bandit feedback, and obtain finite-sample guarantees. 
Then we propose a decentralized algorithm via a novel combination of the Frank-Wolfe method and G-optimal design.
By exploiting the structure of the congestion game, we show the sample complexity of both algorithms
depends only 
polynomially on the number of players and the number of facilities, but not the size of the action set, which can be exponentially large in terms of the number of facilities.
We further define a new problem class, Markov congestion games, which allows us to 
model the non-stationarity in congestion games. 
We propose a centralized algorithm for Markov congestion games, whose sample complexity again has only polynomial dependence on all relevant problem parameters, but not the size of the action set.

\end{abstract}

\setcounter{footnote}{0} 
\section{Introduction}

Nash equilibrium (NE) is a widely adopted concept in game theory community, 
used to describe the behavior of multi-agent systems with selfish players \citep{roughgarden2010algorithmic}. At the Nash equilibrium, no player has the incentive to change its own strategy unilaterally, which implies it is a steady state of the game dynamics. 
For a general-sum game, computing the Nash equilibrium is PPAD-hard \citep{daskalakis2013complexity} and the query complexity is exponential in the number of players \citep{rubinstein2016settling}. To help address these issues, a natural approach is to consider games with special structures. In this paper, we focus on congestion games.

Congestion games are general-sum games with \emph{facilities} (resources) shared among players \citep{rosenthal1973class}. 
During the game, each player will decide what combination of facilities to utilize, and popular facilities will become congested, which results in a possibly higher cost on each user. One example of congestion game is the routing game \citep{fotakis2002structure}, where each player needs to travel from a given starting point to a destination point through some shared routes. These routes are represented as a traffic graph and the facilities are the edges. Each player will decide her path to go, and the more players use the same edge, 
the longer the edge travel time will be.
Congestion games also have wide applications in electrical grids \citep{ibars2010distributed}, internet routing \citep{al2017congestion} and rate allocation \citep{johari2004efficiency}. In many real-world scenarios, players can only have  (semi-)bandit feedback, i.e., players know only the payoff of the facilities they choose. This kind of learning under uncertainty has been widely studied in bandits and in reinforcement learning for the single-agent setting, while theoretical understanding for the multi-agent case is still largely missing. 

There are two types of algorithms in multi-agent systems, namely centralized algorithms and decentralized algorithms.
For centralized algorithms, there exists a central authority that can control and receive feedback from all players in the game. As we have global coordination, centralized algorithms usually have favorable performance. On the other hand, such a central authority may not always be available in practice, and thus people turn to decentralized algorithms, i.e., each player makes decisions individually and can only observe her own feedback. However, decentralized algorithms are vulnerable to \emph{nonstationarity} because each player is making decisions in a nonstationary environment as others' strategies are changing \citep{zhang2021multi}. In this paper, we will study both centralized and decentralized algorithms in congestion games with bandit feedback, and we will provide motivating scenarios for both algorithms in Section \ref{sec:example}.

The main challenge in designing algorithms for $m$-player congestion games with bandit feedback is the curse of exponential action set, i.e., the number of actions can be exponential in the number of facilities $F$ because every subset of facilities can be an action. As a result, an efficient algorithm should have sample complexity polynomial in $m$ and $F$ and has no dependence on the size of the action space. One closely related type of general-sum game is the potential game, in which each individual's payoff changes, resulting from strategy modification, can be quantified by a common potential function.
It is well-known that all congestion games are potential games, and each potential game has an equivalent congestion game formulation \citep{monderer1996potential}. However, existing algorithms designed for potential games all have sample complexity scaling at least linearly in the number of actions \citep{leonardos2021global, ding2022independent}, which is inefficient for congestion games. This motivates the following question:
\begin{center}
    \emph{Can we design provably \revision{sample-efficient} centralized and decentralized \revision{learning} algorithms \revision{for} congestion games with bandit feedback?}
\end{center}
We provide an affirmative answer to this question. \revision{To be precise, we use Nash-regret minimization (formally defined in Section \ref{sec:preliminaries}) as our objective for learning in congestion games. This regret-like objective commonly appears in the literature of online learning and reinforcement learning \citep{orabona2019modern, ding2022independent, liu2021sharp}, which focuses on finite-time analysis and accumulative rewards throughout the learning process instead of the asymptotic behavior. In general, a sublinear Nash regret implies a best-iterate convergence, meaning that the algorithm has reached the approximate Nash equilibrium at least once, while the converse does not hold.
}

We highlight our contributions below and compare our results with previous algorithms in Table \ref{tab:comparisons}. \revision{
Our algorithms are shaded and we prove sublinear Nash regrets for all of them. In Table \ref{tab:comparisons}, 
sample complexity refers to the number of samples required to reach best-iterate convergence to an $\epsilon$-approximate Nash equilibrium and the results are obtained by standard online-to-batch conversion as in Section 3.1 of \citep{jin2018q}. }

\subsection{Main Novelties and Contributions}

\begin{table}[!t]
\centering
\begin{tabular}{ c|c|c|c } 
 \hline
 Algorithms & Sample complexity & Nash regret & Decentralized\\ 
 \hline
 Nash-VI \citep{liu2021sharp} & $(\prod_{i=1}^mA_i)F/\epsilon^2$ & $\sqrt{(\prod_{i=1}^m A_i)FT}$ & No \\ 
 V-learning \citep{jin2021v}& $A_{\max}F/\epsilon^2$ \;(CCE) & NA & Yes \\ 
 IPPG \citep{leonardos2021global} & $A_{\max} mF/\epsilon^6$ & NA & Yes\\
 IPGA \citep{ding2022independent}& $A_{\max}^2m^3F^5/\epsilon^{5}$ & $mF^{4/3}\sqrt{A_{\max}}T^{4/5}$ & Yes\\
 \cellcolor{Gray} Nash-UCB \RNum{1} & $mF^2/\epsilon^2$ & $F\sqrt{mT}$ & No\\
 \cellcolor{Gray} Nash-UCB \RNum{2} & $m^2F^3/\epsilon^2$ & $mF^{3/2}\sqrt{T}$ & No \\
 \cellcolor{Gray} Frank-Wolfe with Exploration \RNum{1}& $m^{12}F^9/\epsilon^6$ & $m^2F^{3/2}T^{5/6}$ & Yes \\
 \cellcolor{Gray} Frank-Wolfe with Exploration \RNum{2}& $m^{12}F^{12}/\epsilon^{6}$ & $m^{2}F^2T^{5/6}$ & Yes \\
 \hline
\end{tabular}
\caption{Comparison of algorithms for congestion games in terms of sample complexity and Nash regret, where ``IPPG'' stands for ``independent projected policy gradient'', ``IPGA'' stands for ``independent policy gradient ascent'', ``\RNum{1}'' represents the setting of semi-bandit feedback and ``\RNum{2}'' represents the setting of bandit feedback. Bandit feedback is assumed for algorithms from previous work. Here, $A_i$ is the size of player $i$'s action space, $m$ is the number of players, $A_{\max}=\max_{i\in[m]}A_i$, $F$ is the number of facilities and $T$ is the number of samples collected.
Our algorithms are shaded.}
\label{tab:comparisons}
\end{table}

\textbf{1. Centralized algorithm for congestion game.} We adapt the principle of optimism in the face of uncertainty in stochastic bandits to ensure sufficient exploration in congestion games. We begin with congestion games with semi-bandit feedback, in which each player can observe the reward of every facility in the action. Instead of estimating the action reward as in stochastic multi-armed bandits, we estimate the facility rewards directly, which \emph{removes the dependence on the size of action space}. Furthermore, we consider congestion games with bandit feedback, in which each player can only observe the overall reward. In this setting, we borrow ideas from linear bandits to estimate the reward function and analyze the algorithm. The algorithm is provably sample efficient in both cases.

\noindent
\textbf{2. Decentralized algorithm for congestion game.} Our decentralized algorithm is a Frank-Wolfe method with exploration, in which each player only observes her own actions and rewards. To efficiently explore in the congestion game, we utilize G-optimal design allocation for bandit feedback and a specific distribution for semi-bandit feedback. As a result, the sample complexity does not depend on the number of actions. In addition, the $L_1$ smoothness parameter of the potential function does not depend on the number of actions, which is exploited by the Frank-Wolfe method. 
With the help of these two specific algorithmic designs for congestion games, 
we give the first decentralized algorithm for both semi-bandit feedback and bandit feedback that has no dependence on the size of the action space in congestion games.

\noindent
\textbf{3. Centralized algorithm for independent Markov congestion game.} We extend the formulation of congestion game into a Markov setting and propose the independent Markov congestion game (IMCG), in which each facility has its own internal state and state transition happens independently among all the facilities. 
In Section~\ref{sec:example}, we give some examples that fit in this model.
By  utilizing techniques from factored MDPs, we extend our centralized algorithms for congestion games to efficiently solve IMCGs, with both semi-bandit and bandit feedback.

\subsection{Motivating Examples}\label{sec:example}

We provide an exmple here to motivate our proposed models. See Section \ref{sec:preliminaries} for the formal definition of (semi-)bandit feedback and (Markov) congestion games and Appendix \ref{sec:additional_examples} for additional examples. 

\begin{example}[{\bf Routing Games}]
For a routing game, there are multiple players in a traffic graph travelling from starting points to destination points, and the facilities are the edges (roads). The cost of each edge is the waiting time, which depends on the number of players using that edge. 

\noindent $\bullet$ \textbf{Centralized algorithm for routing games:} Imagine each player is using Google Maps to navigate. Then Google Maps can serve as a center that knows the starting points and the destination points, as well as the real-time feedback of the waiting time on each edge of all the players. Google Maps itself also has the incentive to assign paths according to the Nash equilibrium strategy as then each player will find out that deviating from the navigation has no benefit and thus sticks to the app. 
    
\noindent $\bullet$ \textbf{Decentralized algorithm for routing games:} Consider the case where players are still using Google Maps but due to privacy concerns or limited bandwidth, they only use the offline version, which has access only to the information of each single user. Then Google Maps needs to use decentralized algorithms so that it can still assign Nash equilibrium strategy to each user after repeated plays.
    
\noindent $\bullet$ \textbf{Markov routing games:} For Markov routing games, the time cost on each edge will change between different timesteps, which is a more accurate model of the real-world. For instance, some roads are prone to car accidents, which will result in an increasing cost on the next timestep, and the chance of accidents also depends on the number of players using that edge currently. This is modeled by the Markovian facility state transition in independent Markov congestion games.

\end{example}

\section{Related Work}
\textbf{Potential Games.}
Potential games are general-sum games that admit a common potential function to quantify the changes in individual's payoff \citep{monderer1996potential}. Algorithmic game theory community has studied how different dynamics converge to the Nash equilibium, e.g., best response dynamics \citep{durand2018analysis,swenson2018best} and no-regret dynamics \citep{heliou2017learning,cheung2020chaos}, while usually they provide only asymptotic convergence, with either full information setting or bandit feedback setting. Recently, reinforcement learning community studied Markov potential games with bandit feedback, which can be applied to standard potential games. See the Markov Games part below for more details.

\noindent
\textbf{Congestion Games.}
Congestion games are developed in the seminal work \citep{rosenthal1973class}, and later \citet{monderer1996potential} builds a close connection between congestion games and potential games. Congestion games are divided into atomic 
and non-atomic congestion games depending on whether each player is separable. Many papers consider non-atomic congestion games with non-decreasing cost function, which implies a convex potential function \citep{roughgarden2004bounding}. We consider the more difficult atomic congestion game where the potential function can be non-convex. For online non-atomic case, \citep{krichene2015online} considers partial information setting while they provide convergence in the sense of Cesaro means. \citep{kleinberg2009multiplicative,krichene2014convergence} show that some no-regret online learning algorithms asymptotically converges to Nash equilibrium. \citep{chen2015playing,chen2016generalized} are two closely related works that consider bandit feedback in atomic congestion games and provide non-asymptotic convergence. However, they still assume a convex potential function and the sample complexity has exponential dependence on the number of facilities, which is far from ideal.

\noindent
\textbf{Markov Games.}
Markov games are widely studied since the seminal work \citep{shapley1953stochastic}. Recently, the topic has received much attention due to advances in reinforcement learning theory. \citet{liu2021sharp} provides a centralized algorithm for learning the Nash equilibrium in general-sum Markov games, and \citep{jin2021v,song2021can} provide decentralized algorithms for learning the (coarse) correlated equilibrium. One closely related line of research is on Markov potential games \citep{leonardos2021global,zhang2021gradient,fox2021independent,cen2022independent,ding2022independent}. However, applying their algorithms to congestion games leads to explicit dependence on the number of actions, which would be exponentially worse than our algorithms.
See Table~\ref{tab:comparisons} for comparisons.
Our independent Markov congestion game is motivated by the state-based potential games studied in \citet{marden2012state} and \citet{macua2018learning}, and its transition kernel is closely related to the factored MDPs, for which single agent algorithms are studied in \citep{osband2014near, chen2020efficient,xu2020reinforcement,tian2020towards,rosenberg2021oracle}.

\noindent
\revision{
\textbf{Learning in Games.} Different from our paper, learning in games in traditional literature of game theory mainly considers players' asymptotic behavior \citep{leslie2005individual, cominetti2010payoff, coucheney2015penalty}. In early literature, \citet{leslie2004reinforcement} investigates actor-critic learning and $Q$-learning algorithms in games with bandit feedback and their connection to best-response dynamics. \citet{leslie2005individual} proposes individual $Q$-learning algorithm and shows that it converges to the NE almost surely in two-player zero-sum game and \citet{leslie2006generalised} studies learning the NE from the perspective of a fictitious play-like process. Later, \citet{cominetti2010payoff} considers payoff-based learning rules and shows convergence to NE in traffic games, while another payoff-based learning model for continuous games is developed in \cite{bervoets2020learning}.  \citet{coucheney2015penalty} derives a new penalty-regulated dynamics and proposes a corresponding learning algorithms that converges to NE in potential games with bandit feedback. \citet{bravo2018bandit} proposes that in monotone games with bandit feedback, as long as all players are using some no-regret learning algorithm, the dynamics will converge to the NE, and an improved analysis of the same derivative-free algorithm is given in \citet{drusvyatskiy2021improved}. In contrast, our learning objective focuses on finite-time cumulative rewards, which is more widely used in current multi-agent reinforcement learning literature \citep{ding2022independent, liu2021sharp}.}

\section{Preliminaries}
\label{sec:preliminaries}

\textbf{General-sum Matrix Games.}
We consider the model of general-sum matrix games, defined by the tuple $\G=\Sp{\Bp{\A_i}_{i=1}^{m}, R}$, where $m$ is the number of players, $\A_i$ is the action space of player $i$ and $R(\cdot|\a)$ is the reward distribution on $[0,r_\mathrm{max}]^m$ with mean $\r(\a)$. Let $\A=\A_1\times\dots\times\A_m$ be the whole action space and denote an element as $\a=(a_1, \dots, a_m)\in\A$. After all players take actions $\a\in\A$, a reward vector is sampled $\r\sim R(\cdot|\a)$ and player $i$ will receive reward $r_i\in[0,r_\mathrm{max}]$ with mean $r_i(\a)$. Each player's objective is to maximize her own reward.

A general policy $\pi$ is defined as a vector in $\Delta(\A)$, the probability simplex over the action space $\A$. A product policy $\pi=(\pi_1, \dots, \pi_m)$ is defined as a tuple in $\Delta(\A_1)\times\dots\times\Delta(\A_m)$, in which $\a=(a_1, \dots, a_m)\sim\pi$ represents $a_i\overset{\mathrm{i.i.d.}}{\sim}\pi_i$. The value of policy $\pi$ for player $i$ is $V_i^\pi=\E_{\a\sim\pi}[r_i(\a)]$.

\noindent
\textbf{Nash Equilibrium and Nash Regret.}
Given a general policy $\pi$, let $\pi_{-i}$ be the marginal joint policy of players $1, \dots, i-1, i+1, \dots, m$. Then, the best response of player $i$ under policy $\pi$ is $\pi_i^{\dagger}=\argmax_{\mu\in\Delta(\A_i)}V_i^{\mu, \pi_{-i}}$ and the corresponding value is $V_i^{\dagger, \pi_{-i}}:=V_i^{\pi_i^\dagger, \pi_{-i}}$. Our goal is to find the approximate Nash equilibrium of the matrix game, which is defined below.
\begin{definition}
A product policy $\pi$ is an $\epsilon$\textit{-approximate Nash equilibrium} if
$\max_{i}(V_i^{\dagger, \pi_{-i}}-V_i^\pi)\leq\epsilon$.
\end{definition}
An $\epsilon$-approximate Nash equilibrium can be obtained by achieving a sublinear Nash regret, which is defined below. See Section 3 in \citet{ding2022independent} for a more detailed discussion. 
\begin{definition}
With $\pi^k$ being the policy at $k$-th episode, the \textit{Nash regret} after $K$ episodes is define as
$$\text{Nash-Regret}(K)=\sum_{k=1}^{K}\max_{i\in[m]}\Sp{V_i^{\dagger, \pikni}-V_i^{\pik}}.$$
\end{definition}
\revision{
\begin{remark}
Here, if we replace $\max_{i\in[m]}$ by $\sum_{i=1}^{m}$ in the definition of Nash regret, the single-step Nash regret at episode $k$ will become the Nikaido-Isoda (NI) function evaluated at $\pi^k$, which is a popular objective for equilibrium computation \citep{nikaido1955note, raghunathan2019game}. Replacing $\max_{i\in[m]}$ by $\sum_{i=1}^{m}$ will 
multiply our regret bounds by a factor of $m$, while our conclusion will not be affected.
\end{remark}
}

\noindent
\textbf{Potential Games.} A potential game is a general-sum game such that there exists a potential function $\Phi:\Delta(\A)\rightarrow[0,\Phi_\mathrm{max}]$ such that for any player $i\in[m]$ and policies $\pi_i$, $\pi_i'$, $\pi_{-i}$, it satisfies
$$\Phi(\pi_i,\pi_{-i})-\Phi(\pi'_i,\pi_{-i})=V_i^{\pi_i,\pi_{-i}}-V_i^{\pi'_i,\pi_{-i}}.$$
We can immediately see that a policy that maximizes the potential function is a Nash equilibrium.


\noindent
\textbf{Congestion Games.}
A congestion game is defined by $\G=(\F, \Bp{\A_i}_{i=1}^m, \Bp{R^f}_{f\in\mc{F}})$, where $\F=[F]$ is called the facility set and $R^f(\cdot|n)\in [0, 1]$ is the reward distribution for facility $f$ with mean $r^f(n)$, where $n\in[m]$. 
Each action $a_i\in\A_i$ is a subset of $\F$ (i.e., $a_i\subseteq\F$). Suppose the joint action chosen by all the players is $\a\in\A$, then a random reward is sampled $r^f\sim R^f(\cdot|n^f(\a))$ for each facility $f$, where $n^f(\a)=\sum_{i=1}^m\mathds{1}\Bp{f\in a_i}$ is the number of players using facility $f$. The reward collected by player $i$ is $r_i=\sum_{f\in a_i}r^f$ with mean $r_i(\a)=\sum_{f\in a_i}r^f(n^f(\a))\in[0,F]$.

\noindent
\textbf{Connection to Potential Games \citep{monderer1996potential}.}
As a special class of potential game, all congestion games have the potential function: $\Phi(\a)=\sum_{f\in\F}\sum_{i=1}^{n^f(\a)}r^f(i)$.
To see this, we can easily verify that
$\Phi(a_i,a_{-i})-\Phi(a'_i,a_{-i})=r_i(a_i,a_{-i})-r_i(a'_i,a_{-i})$ holds. Then, by defining $\Phi(\pi)=\E_{\a\sim\pi}[\Phi(\a)]$, we can have
$\Phi(\pi_i,\pi_{-i})-\Phi(\pi'_i,\pi_{-i})=V^{\pi_i,\pi_{-i}}_i-V^{\pi'_i,\pi_{-i}}_i$.

\noindent
\textbf{Types of feedback.}
There are in general two types of reward feedback for the congestion games, semi-bandit feedback and bandit feedback, both of which are reasonable under different scenarios. In semi-bandit feedback, after taking the action, player $i$ will receive reward information $r^f$ for each $f\in a_i$; in bandit feedback, after taking the action, player $i$ will only receive the reward $r_i=\sum_{f\in a_i}r^f$ with no knowledge about each $r^f$. In this paper, we will address both of them, with more focus on the bandit feedback, which can be directly generalized to semi-bandit feedback.

\section{Centralized Algorithms for Congestion Games}
\label{sec:centralized_algo_mg}
In this section, we introduce two centralized algorithms for congestion games -- one for the semi-bandit feedback and one for the bandit feedback. We will see that both of them can achieve sublinear Nash regret with polynomial dependence on both $m$ and $F$.

\subsection{Algorithm for Semi-bandit Feedback}
Summarized in Algorithm \ref{algo:nash_vi_matrix_game}, Nash upper confidence bound (Nash-UCB) for congestion games is developed based on optimism in the face of uncertainty. In particular, the algorithm estimates the reward matrices optimistically in line \ref{line:nash_vi_est}, computes its Nash equilibrium policy in line \ref{line:compute_nash_mg} and then follows this policy. 

For convenience, we define the empirical counter $N^{k, f}(n)=\sum_{k'=1}^{k}\mathds{1}\Bp{n^f(\bm{a}^{k'})=n}$ and $\Tilde{\iota} = 2\log(4(m+1)K/\delta)$. Then, the reward estimator for $f$ and the bonus term are defined as
\begin{equation}
    \label{equ:r_estimate_semi}
    \hat{r}^{k, f}(n)=\frac{\sum_{k'=1}^k r^{k', f}\mathds{1}\Bp{n^f(\bm{a}^{k'})=n}}{N^{k, f}(n)\vee 1},\quad b_i^{k, \reward}(\a)=\sum_{f\in a_i}\sqrt{\frac{\tilde{\iota}}{N^{k, f}(n^f(\a))\vee 1}},
\end{equation}
where $r^{k, f}\in[0, 1]$ is the random reward realization of $r^f(n^f(\a^k))$. Naturally, the reward estimator for player $i$ is $\hat{r}^k_i(\a)=\sum_{f\in a_i}\hat{r}^{k, f}(n^f(\a))$.

\begin{algorithm}[ht]
    \caption{Nash-UCB for Congestion Games}
    \label{algo:nash_vi_matrix_game}
    \begin{algorithmic}[1]
    \STATE {\bf Input:} $\epsilon$, accuracy parameter for Nash equilibrium computation
    \FOR{episode $k=1, \dots, K$}
        \FOR{player $i=1, \dots, m$}
            \STATE $\overline{Q}^k_{i}(\bm{a})\leftarrow\hat{r}^k_{i}(\a)+b_i^{k, \reward}(\bm{a})$ for all $\bm{a}\in\mc{A}$\label{line:nash_vi_est}
        \ENDFOR
        \STATE $\pi^k\leftarrow $ $\epsilon$-$\textsc{Nash}(\overline{Q}^k_{1}(\cdot), \cdots, \overline{Q}^k_{m}(\cdot))$ (Algorithm \ref{algo:eps NE}) \label{line:compute_nash_mg}
        \STATE Take action $\bm{a}^k\sim\pi^k$ and observe reward $r^{k, f}$
        \STATE Update reward estimators $\hat{r}^{k}_i$ and bonus term $b_i^{k, \reward}$
    \ENDFOR
    \end{algorithmic}
\end{algorithm}

Algorithm \ref{algo:nash_vi_matrix_game} is motivated by the Nash-VI algorithm in \citep{liu2021sharp} plus a deliberate utilization of the special reward structure in the congestion games. Moreover, notice that a matrix game with reward functions $\oq^k_1(\cdot), \dots, \oq^k_m(\cdot)$ forms a potential game (see Lemma \ref{lmm:q_potential_game}). As a result, in line \ref{line:compute_nash_mg}, we can \emph{efficiently compute} the $\epsilon$-approximate Nash equilibrium $\pi^k$ for that matrix game by utilizing Algorithm \ref{algo:eps NE}, (see Lemma \ref{lmm:compute_epsilon_nash}). It is a simple greedy algorithm such that in each round, it modifies one player's policy whose modification can increase the potential function most. In addition, Algorithm \ref{algo:eps NE} always outputs a deterministic product policy.

\begin{algorithm}[ht]
    \caption{$\epsilon$-approximate Nash Equilibrium for Potential Games}
    \label{algo:eps NE}
    \begin{algorithmic}[1]
    \STATE {\bf Input:} $\epsilon$, accuracy parameter; full information potential game $\Sp{\Bp{\A_i}_{i=1}^{m}, \Bp{r_i}_{i=1}^{m}}$ such that $r_i\in[0,r_\mathrm{max}]$ for all $i\in[m]$
    \STATE {\bf Initialize:} $\pi^1=\a^1$, arbitrary deterministic product policy
    \FOR{round $k=1, \dots, \left\lceil\frac{mr_\mathrm{max}}{\epsilon}\right\rceil$}
        \FOR{player $i=1, \dots, m$}
            \STATE $\Delta_i=\max_{a_i\in\A_i} r_i(a_i,\pi_{-i}^k)-r_i(\pi^k)$
            \STATE $a_i^{k+1}=\argmax_{a\in\A_i} r_i(a_i,\pi_{-i}^k)-r_i(\pi^k)$
        \ENDFOR
        \IF{$\max_{i\in[m]}\Delta_i\leq\epsilon$}
            \RETURN $\pi^k$
        \ENDIF
        \STATE $j=\argmax_{i\in[m]}\Delta_i$
        \STATE $\pi^{k+1}(j)=a_j^{k+1}$, $\pi^{k+1}(i)=\pi^{k}(i)$, for all $i\neq j$
    \ENDFOR
    \end{algorithmic}
\end{algorithm}

\subsection{Algorithm for Bandit Feedback}\label{sec:centralized bandit}
When the players can only receive bandit feedback, estimating $\hat{r}^{k, f}$ directly for each $f\in\mc{F}$ is no longer feasible. However, notice that the reward function $r_i(\a)=\sum_{f\in a_i}r^f(n^f(\a))$ can be seen as an inner product between vectors characterized by action $\a$ and reward function $r^f(\cdot)$. Therefore, under bandit feedback, we can treat it as a linear bandit and use ridge regression to build the reward estimator $\tilde{r}^{k}_i$ and corresponding bonus term $\tilde{b}^{k, \reward}$, whose index $i$ is dropped since it is the same for all players. The new algorithm will use these two terms to replace $\hat{r}^k_i$ and $b^{k, \reward}_i$ in line \ref{line:nash_vi_est} of Algorithm \ref{algo:nash_vi_matrix_game}.

In particular, define $\theta\in[0, 1]^{\td}$ with $\td=mF$ to be the vector such that $r^f(n)=\theta_{n+m(f-1)}$. Meanwhile, for player $i\in[m]$, define $A_i:\mc{A}\mapsto\Bp{0, 1}^{\td}$ to be the vector-valued function such that 
$$[A_i(\a)]_j=\mathds{1}\Bp{j=n+m(f-1), f\in a_i, n=n^f(\a)}.$$
In other words, $A_i(\a)$ is a 0-1 vector with element 1 only at indices corresponding to those in $\theta$ that represents $r^f(n)$ for $f\in a_i$ and $n=n^f(\a)$. Now, with these definitions, the reward function can be written as $r_i(\a)=\inner{A_i(\a), \theta}$. Then, we build the reward estimator and the bonus term through ridge regression and corresponding confidence bound, which are defined as the following:
\begin{align}
    \tilde{r}^k_{i}(\a)=\inner{A_i(\a), \widehat{\theta}^k},\quad \tilde{b}^{k, \reward}(\a)=\max_{i\in[m]}\Norm{A_i(\a)}_{\Sp{V^k}^{-1}}\sqrt{\tilde{\beta}_k},\label{equ:r_estimate_bandit}
\end{align}
where $\widehat{\theta}^k=\Sp{V^k}^{-1}\sum_{k'=1}^{k-1}\sum_{i=1}^{m}A_i(\a^{k'})r_{i}^{k'}$, $V^k=I+\sum_{k'=1}^{k-1}\sum_{i=1}^{m}A_i(\a^{k'})A_i(\a^{k'})^\top$ and $\sqrt{\tilde{\beta}_k}=\sqrt{\td}+\sqrt{F\td\log\Sp{1+\frac{mkF}{\td}}+F\tilde{\iota}}$. 
Note that we cannot bound the sum of this bonus terms by directly applying the elliptical potential lemma. We instead prove its variant in Lemma \ref{lmm:elliptical_potential_variant}.

\subsection{Regret Analysis}
The Nash regret bounds for the two versions of Algorithm \ref{algo:nash_vi_matrix_game} are formally presented in Theorem \ref{theo:ne_mg}. The proof details are deferred to Appendix \ref{sec:ne_proof_mg}.

\begin{theorem}
\label{theo:ne_mg}
Let $\epsilon=1/K$. For congestion games with semi-bandit feedback, by running Algorithm \ref{algo:nash_vi_matrix_game} with reward estimator and bonus term in \eqref{equ:r_estimate_semi}, with probability at least $1-\delta$, we can achieve that
$$\text{Nash-Regret}(K)\leq\tmco\Sp{F\sqrt{mK}}.$$
Furthermore, if we only have bandit feedback, then by running Algorithm \ref{algo:nash_vi_matrix_game} with reward estimator and bonus term in \eqref{equ:r_estimate_bandit}, with probability at least $1-\delta$, we can achieve that
$$\text{Nash-Regret}(K)\leq\tmco\Sp{mF^{3/2}\sqrt{K}}.$$
\end{theorem}

\begin{remark}
Since each action is a subset of $\mc{F}$, the size of each player's action space can be $2^F$. As a result, directly applying Nash-VI in \citep{liu2021sharp} leads to a regret bound exponential in $F$. 
\end{remark}

\begin{remark}
Note that we assume $r^f\in[0, 1]$, which implies $r_i\in[0, F]$ for each player $i\in[m]$.
\end{remark}
\section{Decentralized Algorithms for Congestion Games}
\label{sec:decentralized_algo}

In this section, we present a decentralized algorithm for congestion games. Due to limited space, we only introduce the version of bandit feedback as in Section \ref{sec:centralized bandit}. The algorithmic details for the semi-bandit feedback setting are deferred into Appendix \ref{sec:fw_semi_bandit}. We will show that under both settings, even though each player can only observe her own actions and rewards, our decentralized algorithm still enjoys sublinear Nash regret with polynomial dependence on $m$ and $F$.

We first define the vector-valued function $\phi_i:\A_i\mapsto\Bp{0, 1}^{F_i}$ to be the feature map of player $i$ such that $[\phi_i(a_i)]_f=\mathds{1}\Bp{f\in a_i}$ for $a_i\in\A_i$ and $f\in\bigcup_{a_i\in\A_i}a_i$. Here, $F_i$ is the size of $\bigcup_{a_i\in\A_i}a_i\subseteq\F$ and we can immediately see that $F_i\leq F$ for any $i\in[m]$. 

The core idea of our algorithm is that the Nash equilibrium can be found by reaching the stationary points of the potential function since all congestion games are potential games. \revision{Here, the UCB-like algorithms used in the centralized setting are not applicable because their policy computation requires value functions for all players (e.g., line \ref{line:compute_nash_mg} of Algorithm \ref{algo:nash_vi_matrix_game}), which are not available in the decentralized setting.} Summarized in Algorithm \ref{algo:FW}, the decentralized algorithm is developed based on the Frank-Wolfe method and has the following three major components.

\begin{algorithm}[!t]
    \caption{Frank-Wolfe with Exploration for Congestion Game}
    \label{algo:FW}
    \begin{algorithmic}[1]
    \STATE {\bf Input:} $\gamma, \nu$, mixture weights; $\pi_i^1$, initial policy.
    \STATE {\bf Initialize:} $\rho_i$, the G-optimal design for player $i$, defined in \eqref{equ:G_optimal}. 
    \FOR{episode $k=1,\cdots,K$}
        \FOR{round $t=1,\cdots,\tau$}
            \STATE Each player takes action $a_i^{k,t}\sim\pi_i^{k}$, observes reward $r^{k,t}_i$. \label{line:fw_sample}
        \ENDFOR
        \FOR{player $i=1,\cdots,m$}
            \STATE Compute $\widehat{\nabla}^k_i\Phi(a_i)$ by the formula in \eqref{equ:grad_estimator} for all $a_i\in\A_i$ \label{line:fw_gradient}
            \STATE Compute $\widetilde{\pi}_i^{k+1}\leftarrow \argmax_{\pi_i\in\Delta(\A_i)} \inner{\pi_i,\widehat{\nabla}^{k}_i\Phi}$ \label{line:fw_argmax}
            \STATE Update $\pi_i^{k+1}\leftarrow (1-\gamma)(\nu\widetilde{\pi}_i^{k+1}+(1-\nu)\pi_i^{k})+\gamma \rho_i$ \label{line:fw_update}
        \ENDFOR
    \ENDFOR
    \end{algorithmic}
\end{algorithm}

\paragraph{Gradient Estimator.}
In line \ref{line:fw_gradient}, the algorithm builds the estimator $\widehat{\nabla}^k_i\Phi$ defined in \eqref{equ:grad_estimator} by using the $\tau$ reward samples collected from line \ref{line:fw_sample}. Here, $\widehat{\nabla}^k_i\Phi$ estimates the gradient of potential function $\Phi$ with respect to the policy $\pi_i^k$. 
Recall that for a congestion game, we have $\Phi(\a)=\sum_{f\in\F}\sum_{i=1}^{n^f(\a)}r^f(i)$ and $\Phi(\pi)=\E_{\a\sim\pi}\Mp{\Phi(\a)}$. Then we can define $\nabla_i\Phi:=\nabla_{\pi_i}\Phi$ as a vector of dimension $\abs{\A_i}$. For the component indexed by some $a_i\in\A_i$, we can see that $\Phi(\pi)=\pi_i(a_i)\E_{a_{-i}\sim\pi_{-i}}\Mp{r_i(a_i, a_{-i})}+\mathrm{const}$, where const does not depend on $\pi_i(a_i)$. Therefore, we have
\begin{equation}
    \label{equ:phi_grad}
    \nabla_i\Phi(a_i)=\E_{a_{-i}\sim\pi_{-i}}\Mp{r_i(a_i, a_{-i})}=\E_{a_{-i}\sim\pi_{-i}}\Mp{\sum_{f\in a_i}r^f(n^f(a_i, a_{-i}))}=\inner{\phi_i(a_i), \theta_i(\pi)},
\end{equation}
where $[\theta_i(\pi)]_f=\E_{a_{-i}\sim\pi_{-i}}\Mp{r^f(n^f(a_{-i})+1)}$. Meanwhile, the mean of the $t$-th reward that player $i$ received at episode $k$ satisfies
$$\E\Mp{r_i^{k, t}\mid \a^{k, t}}=r_i(\a^{k,t})=\sum_{f\in a_i^{k, t}}r^{f}(n^f(\a^{k,t}))=\inner{\phi_i(a_i^{k,t}),\theta^{k,t}_i(a^{k, t}_{-i})},$$
where $[\theta^{k,t}_i(a^{k, t}_{-i})]_f=r^{f}(n^f(a^{k,t}_{-i})+1)$ and its mean is $[\theta_i(\pi^k)]_f$. Therefore, we can use linear regression to estimate $\theta_i(\pi^k)$. In particular, we have
$\widehat{\theta}^k_i(\pi^k)=\frac{1}{\tau}\sum_{t=1}^{\tau}\Sp{\Sigma^k_i}^{-1}\phi_i(a_i^{k, t})r_i^{k, t},$
with the covariance matrix $\Sigma^k_i=\E_{a_i\sim\pi_i^k}\Mp{\phi_i(a_i)\phi_i(a_i)^\top}$.
Then, we have the unbiased gradient estimate
\begin{equation}
    \label{equ:grad_estimator}
    \widehat{\nabla}^k_i\Phi(a_i)=\inner{\phi_i(a_i), \widehat{\theta}_i^k(\pi^k)}=\frac{1}{\tau}\sum_{t=1}^{\tau}\phi_i(a_i)^\top\Sp{\Sigma^k_i}^{-1}\phi_i(a_i^{k, t})r_i^{k, t}.
\end{equation}

\revision{
\begin{remark}
One difference between Algorithm \ref{algo:FW} (decentralized) and Algorithm \ref{algo:nash_vi_matrix_game} (centralized) is that in the decentralized algorithm, each player is required to play the same policy for $\tau$ times before an update can be applied. An episode is thus defined for convenience as the time period during which the players' policies are fixed. We make this artificial design mainly for controlling the variance of the gradient estimator $\widehat{\nabla}^k_i\Phi(a_i)$. However, we conjecture that with more careful design and analysis, it should be possible to improve Algorithm \ref{algo:FW} so that only one sample is required per episode \citep{zhang2020one}.
\end{remark}
}
 
\paragraph{G-optimal Design.}   
In line \ref{line:fw_argmax} and \ref{line:fw_update}, the algorithm performs standard Frank-Wolfe update and mixes the updated policy with an exploration policy $\rho_i$, which is defined as the G-optimal allocation for features $\Bp{\phi_i(a_i)}_{a_i\in\A_i}$. To be specific, we have
\begin{equation}
    \label{equ:G_optimal}
    \rho_i=\argmin_{\lambda\in\Delta(\A_i)}\max_{a_i\in\A_i}\Norm{\phi_i(a_i)}^2_{\E_{a_i'\sim\lambda}\Mp{\phi_i(a_i')\phi_i(a_i')^\top}^{-1}}.
\end{equation}
Here $\rho_i$ guarantees that $\Sigma^k_i$ is invertible and the variance of $\widehat{\nabla}_i^k\Phi(a_i)=\inner{\phi_i(a_i), \widehat{\theta}_i^k(\pi^k)}$ depends only on $F$ instead of the size of action space (Lemma \ref{lemma:variance bound}) because by the famous Kiefer-Wolfowitz theorem, we have $\max_{a_i\in\A_i}\Norm{\phi_i(a_i)}^2_{\E_{a_i'\sim\rho_i}\Mp{\phi_i(a_i')\phi_i(a_i')^\top}^{-1}}=F_i\leq F$ \citep{lattimore2020bandit}.

\paragraph{Frank-Wolfe Update.}   

Finally, we emphasize that it is crucial to use Frank-Wolfe update because it is compatible with \emph{$L_1$ norm} and we can show that $\Phi$ is $mF$-smooth with respect to the $L_1$ norm (Lemma \ref{lmm:smooth}). In contrast, its smoothness for $L_2$ norm will depend on the size of the action space.

Before the game starts, each player $i$ can compute her $\rho_i$ based on her own action set $\A_i$. During the game, all players only have access to their own actions and rewards, which means that Algorithm \ref{algo:FW} is fully decentralized. The Nash regret bound for this algorithm is formally stated in Theorem \ref{theo:dec_ne} and the proof details are given in Appendix \ref{sec:explore_dist} and \ref{sec:fw_proof}.

\begin{theorem}
\label{theo:dec_ne}
Let $T=K\tau$. For congestion game with bandit feedback, by running Algorithm \ref{algo:FW} with gradient estimator $\widehat{\nabla}^k_i\Phi$ in \eqref{equ:grad_estimator} and exploration distribution $\rho_i$ in \eqref{equ:G_optimal}, if $K\geq\frac{2F}{m}$, then with probability at least $1-\delta$, we have
$$\text{Nash-Regret}(T):= \sum_{k=1}^{K}\tau\max_{i\in[m]}\Sp{V_i^{\dagger, \pikni}-V_i^{\pik}} \leq\tmco\Sp{m^2F^2T^{5/6}+m^3F^3T^{2/3}}.$$
For congestion game with semi-bandit feedback, by running Algorithm \ref{algo:FW} with gradient estimator $\widetilde{\nabla}^k_i\Phi(a_i)$ and exploration distribution $\tilde{\rho}_i$ defined in Appendix \ref{sec:fw_semi_bandit}, if $K\geq\frac{2\sqrt{F}}{m}$, then with probability at least $1-\delta$, we have
$$\text{Nash-Regret}(T) \leq\tmco\Sp{m^2F^{3/2}T^{5/6}+m^3F^2T^{2/3}}.$$
\end{theorem}

\section{Extension to Independent Markov Congestion Games}
\label{sec:centralized_algo_markov}

In this section, we propose and analyze a Markov extension of the congestion games, called the independent Markov congestion games (IMCGs). 

\subsection{Problem Formulation}

\textbf{General-sum Markov Games.} A finite-horizon time-inhomogeneous tabular general-sum Markov game is defined by $\M=\{\S,\Bp{\A_i}_{i=1}^m,H,P,R, s_0\}$, where $\S$ is the state space, $m$ is the number of players, $\A_i$ is the action space of player $i$, $\A=\A_1\times\dots\times\A_m$ is the whole action space, $H$ is the time horizon, $s_0$ is the initial state\footnote{An episode is defined as running $H$ steps from the initial state $s_0$, which is common for the episodic MDP.}, $P=(P_1,P_2,\cdots,P_H)$ with $P_h\in[0, 1]^{S\times A\times S}$ as the transition kernel at timestep $h$, $R=\{R_{h}(\cdot|s_h,\a_h)\}_{h=1}^{H}$ with $R_h(\cdot|s_h,\a_h)$ as the reward distribution on $[0,r_\mathrm{max}]^m$ with mean $\r_h(s_h,a_h)\in[0,r_\mathrm{max}]^m$ at timestep $h\in[H]$. At timestep $h$, all players choose their actions simultaneously and a reward vector is sampled $\r_h\sim R_h(\cdot|s_h,\a_h)$, where $s_h$ is the current state and $\a_h=(a_{h,1},a_{h,2},\cdots,a_{h,m})$ is the joint action. Each player $i$ receives reward $r_{h,i}$ and the state transits to $s_{h+1}\sim P_h(\cdot|s_h,\a_h)$. The objective for each player is to maximize her own total reward. We assume that the initial state $s_1$ is fixed.
 
A (Markov) policy $\pi$ is a collection of $H$ functions $\Bp{\pi_h:\S\mapsto\Delta(\A)}_{h=1}^{H}$, each of which maps a state to a distribution over the action space. $\pi$ is a product policy if $\pi_h(\cdot\mid s)$ is a product policy for each $(h, s)\in[H]\times\S$. The value function and $Q$-value function of player $i$ at timestep $h$ under policy $\pi$ are defined as
\fontsize{9}{9}
\begin{align*}
    V^{\pi}_{h, i}(s)=\E_{\pi}\Mp{\sum_{h'=h}^{H}r_{h', i}(s_{h'}, \a_{h'})\mid s_h=s},\  Q^{\pi}_{h, i}(s, \a)=\E_{\pi}\Mp{\sum_{h'=h}^{H}r_{h', i}(s_{h'}, \a_{h'})\mid s_h=s, \a_h=a}.
\end{align*}
\normalsize

The best responses and Nash regret can be defined similarly as those for matrix games. In particular, given a policy $\pi$, player $i$'s best response policy is $\pi^\dagger_{h, i}(\cdot\mid s)=\argmax_{\mu\in\Delta(\A_i)}V^{\mu, \pi_{-i}}_{h, i}(s)$ and the corresponding value function is denoted as $V^{\dagger, \pi_{-i}}_{h, i}$.

\begin{definition}
With $\pi^k$ being the policy at $k$th episode, the \textit{Nash regret} after $K$ episodes is define as
$$\text{Nash-Regret}(K)=\sum_{k=1}^{K}\max_{i\in[m]}\Sp{V_{1, i}^{\dagger, \pikni}-V_{1, i}^{\pik}}(s_1).$$
\end{definition}


\noindent
\textbf{Independent Markov Congestion Game.}
A general-sum Markov game is an independent Markov congestion game (IMCG) if there exists a facility set $\F$ such that $a_i\subseteq\F$ for any $a_i\in\A_i$, a state space $\S=\prod_{f\in\F}\S^f$, a set of facility reward distributions $\{R_h^f\}_{h\in[H],f\in\F}$ such that if the joint action at $s_h$ is $\a$, we have
$r_{h,i}=\sum_{f\in a_i}r_h^f$, 
where $r_h^f\sim R_h^f(\cdot|s_h,n^f(\a))$ with support on $[0,1]$ and mean $r_h^f(s_h,n^f(\a))$, and a set of transition matrices $\{P^f_h\}_{h\in[H], f\in\F}$ such that
$P_h(s'|s,\a)=\prod_{f\in\F}P^f_h(s'^f|s^f,n^f(\a))$.
In other words, at each timestep $h$ and state $s\in\S$, the players are in a congestion game. Meanwhile, each facility has its own state and independent state transition, which only depends on its current state and number of players using that facility. This transition kernel can be viewed as a special case of that in factored MDPs \citep{szita2009optimistic}. The IMCG also admits two types of feedback, semi-bandit feedback and bandit feedback, just like the congestion game. In this paper, we will consider both types of feedback.

\subsection{Theoretical Guarantee}

Summarized in Algorithm \ref{algo:nash_vi}, our centralized algorithm for IMCGs is naturally extended from the Nash-UCB (Algorithm \ref{algo:nash_vi_matrix_game}) by incorporating transition kernel estimators, corresponding bonus terms and Bellman backward update. The key idea is to utilize the independent transition structure to remove the dependence on the exponential size of the state space $S=\prod_{f\in\F}S^f$. We tackle this issue by adapting technique from factored MDP \citep{chen2020efficient}. The algorithmic details for both types of feedback are deferred into Appendix \ref{sec:markov_algo_details}. The Nash regret bounds for the two versions of Algorithm \ref{algo:nash_vi} are stated in Theorem \ref{theo:ne} and the proof details are deferred to Appendix \ref{sec:ne_proof}.

\begin{theorem}
\label{theo:ne}
For independent Markov congestion game with semi-bandit feedback, by running the centralized Algorithm \ref{algo:nash_vi}, with probability at least $1-\delta$, we can achieve that
$$\text{Nash-Regret}(K)\leq \tmco\Sp{\sum_{f\in\mc{F}}FS^f\sqrt{mH^3T}} + \tmco\Sp{m^2H^2F\sum_{f\neq f'}\Sp{S^fS^{f'}}^2}.$$

Furthermore, if we only have bandit feedback, then by running Algorithm \ref{algo:nash_vi} with reward estimator and bonus term in \eqref{equ:r_tilde} and \eqref{equ:reward_bonus_2}, with probability at least $1-\delta$, we can achieve that
$$\text{Nash-Regret}(K)\leq \tmco\Sp{\sum_{f\in\mc{F}}FS^f\sqrt{m^2H^3T}} + \tmco\Sp{m^2H^2F\sum_{f\neq f'}\Sp{S^fS^{f'}}^2}.$$
\end{theorem}
The regret bound in \citep{liu2021sharp} is $\widetilde{O}(\sqrt{H^3S^2(\Pi_{i=1}^mA_i)T})$, where both $A_i$ and $S=\prod_{f\in\F}S^f$ can be exponential in $F$. Our bounds have polynomial dependence on all the parameters. 


\section{Conclusion}
\label{sec:conclusion}

In this paper, we study sample-efficient learning in congestion games by utilizing the special reward structure. We propose both centralized and decentralized algorithms for congestion games with two types of feedback, 
all achieving sample complexities only polynomial in the number of facilities. 
To the best of our knowledge, each one of them is the first sample-efficient learning algorithm for congestion games in its own setting. 
We further define the independent Markov congestion game (IMCG) as a natural extension of the congestion game into the Markov setting together with a sample-efficient centralized algorithm for both types of feedback. 

One promising future direction is to find a sample-efficient decentralized algorithm such that from each player's own perspective, the algorithm is still no-regret. In other words, diminishing regret is guaranteed for the player by running this algorithm even though other players may use policies from different algorithms. Another important future direction is to find \revision{sample-efficient centralized/decentralized algorithms that can explicitly find an approximate Nash equilibrium policy.}
\section*{Acknowledgements}

We sincerely thank Jing Dong for pointing out a mistake in the initial draft of this paper. This work was supported in part by NSF TRIPODS II-DMS 2023166, NSF CCF 2007036, NSF IIS 2110170, NSF DMS 2134106, NSF CCF 2212261, NSF IIS 2143493, NSF CCF 2019844.

\bibliography{References}

\begin{thebibliography}{53}
\providecommand{\natexlab}[1]{#1}
\providecommand{\url}[1]{\texttt{#1}}
\expandafter\ifx\csname urlstyle\endcsname\relax
  \providecommand{\doi}[1]{doi: #1}\else
  \providecommand{\doi}{doi: \begingroup \urlstyle{rm}\Url}\fi

\bibitem[Al-Kashoash et~al.(2017)Al-Kashoash, Hafeez, and
  Kemp]{al2017congestion}
Hayder~AA Al-Kashoash, Maryam Hafeez, and Andrew~H Kemp.
\newblock Congestion control for 6lowpan networks: A game theoretic framework.
\newblock \emph{IEEE internet of things journal}, 4\penalty0 (3):\penalty0
  760--771, 2017.

\bibitem[Bai and Jin(2020)]{bai2020provable}
Yu~Bai and Chi Jin.
\newblock Provable self-play algorithms for competitive reinforcement learning.
\newblock In \emph{International conference on machine learning}, pages
  551--560. PMLR, 2020.

\bibitem[Bervoets et~al.(2020)Bervoets, Bravo, and Faure]{bervoets2020learning}
Sebastian Bervoets, Mario Bravo, and Mathieu Faure.
\newblock Learning with minimal information in continuous games.
\newblock \emph{Theoretical Economics}, 15\penalty0 (4):\penalty0 1471--1508,
  2020.

\bibitem[Bravo et~al.(2018)Bravo, Leslie, and Mertikopoulos]{bravo2018bandit}
Mario Bravo, David Leslie, and Panayotis Mertikopoulos.
\newblock Bandit learning in concave n-person games.
\newblock \emph{Advances in Neural Information Processing Systems}, 31, 2018.

\bibitem[Cen et~al.(2022)Cen, Chen, and Chi]{cen2022independent}
Shicong Cen, Fan Chen, and Yuejie Chi.
\newblock Independent natural policy gradient methods for potential games:
  Finite-time global convergence with entropy regularization.
\newblock \emph{arXiv preprint arXiv:2204.05466}, 2022.

\bibitem[Chen and Lu(2015)]{chen2015playing}
Po-An Chen and Chi-Jen Lu.
\newblock Playing congestion games with bandit feedbacks.
\newblock In \emph{AAMAS}, pages 1721--1722, 2015.

\bibitem[Chen and Lu(2016)]{chen2016generalized}
Po-An Chen and Chi-Jen Lu.
\newblock Generalized mirror descents in congestion games.
\newblock \emph{Artificial Intelligence}, 241:\penalty0 217--243, 2016.

\bibitem[Chen et~al.(2020)Chen, Hu, Li, and Wang]{chen2020efficient}
Xiaoyu Chen, Jiachen Hu, Lihong Li, and Liwei Wang.
\newblock Efficient reinforcement learning in factored mdps with application to
  constrained rl.
\newblock \emph{arXiv preprint arXiv:2008.13319}, 2020.

\bibitem[Cheung and Piliouras(2020)]{cheung2020chaos}
Yun~Kuen Cheung and Georgios Piliouras.
\newblock Chaos, extremism and optimism: Volume analysis of learning in games.
\newblock \emph{Advances in Neural Information Processing Systems},
  33:\penalty0 9039--9049, 2020.

\bibitem[Cominetti et~al.(2010)Cominetti, Melo, and Sorin]{cominetti2010payoff}
Roberto Cominetti, Emerson Melo, and Sylvain Sorin.
\newblock A payoff-based learning procedure and its application to traffic
  games.
\newblock \emph{Games and Economic Behavior}, 70\penalty0 (1):\penalty0 71--83,
  2010.

\bibitem[Coucheney et~al.(2015)Coucheney, Gaujal, and
  Mertikopoulos]{coucheney2015penalty}
Pierre Coucheney, Bruno Gaujal, and Panayotis Mertikopoulos.
\newblock Penalty-regulated dynamics and robust learning procedures in games.
\newblock \emph{Mathematics of Operations Research}, 40\penalty0 (3):\penalty0
  611--633, 2015.

\bibitem[Daskalakis(2013)]{daskalakis2013complexity}
Constantinos Daskalakis.
\newblock On the complexity of approximating a nash equilibrium.
\newblock \emph{ACM Transactions on Algorithms (TALG)}, 9\penalty0
  (3):\penalty0 1--35, 2013.

\bibitem[Ding et~al.(2022)Ding, Wei, Zhang, and
  Jovanović]{ding2022independent}
Dongsheng Ding, Chen-Yu Wei, Kaiqing Zhang, and Mihailo~R. Jovanović.
\newblock Independent policy gradient for large-scale markov potential games:
  Sharper rates, function approximation, and game-agnostic convergence, 2022.

\bibitem[Drusvyatskiy et~al.(2022)Drusvyatskiy, Fazel, and
  Ratliff]{drusvyatskiy2021improved}
Dmitriy Drusvyatskiy, Maryam Fazel, and Lillian~J Ratliff.
\newblock Improved rates for derivative free gradient play in strongly monotone
  games.
\newblock In \emph{Proc. {IEEE} Conference on Decision and Control}, 2022.

\bibitem[Durand(2018)]{durand2018analysis}
Stéphane Durand.
\newblock \emph{Analysis of Best Response Dynamics in Potential Games}.
\newblock PhD thesis, Universit{\'e} Grenoble Alpes, 2018.

\bibitem[Fotakis et~al.(2002)Fotakis, Kontogiannis, Koutsoupias, Mavronicolas,
  and Spirakis]{fotakis2002structure}
Dimitris Fotakis, Spyros Kontogiannis, Elias Koutsoupias, Marios Mavronicolas,
  and Paul Spirakis.
\newblock The structure and complexity of nash equilibria for a selfish routing
  game.
\newblock In \emph{International Colloquium on Automata, Languages, and
  Programming}, pages 123--134. Springer, 2002.

\bibitem[Fox et~al.(2021)Fox, McAleer, Overman, and
  Panageas]{fox2021independent}
Roy Fox, Stephen McAleer, Will Overman, and Ioannis Panageas.
\newblock Independent natural policy gradient always converges in markov
  potential games.
\newblock \emph{arXiv preprint arXiv:2110.10614}, 2021.

\bibitem[Heliou et~al.(2017)Heliou, Cohen, and
  Mertikopoulos]{heliou2017learning}
Am{\'e}lie Heliou, Johanne Cohen, and Panayotis Mertikopoulos.
\newblock Learning with bandit feedback in potential games.
\newblock \emph{Advances in Neural Information Processing Systems}, 30, 2017.

\bibitem[Ibars et~al.(2010)Ibars, Navarro, and Giupponi]{ibars2010distributed}
Christian Ibars, Monica Navarro, and Lorenza Giupponi.
\newblock Distributed demand management in smart grid with a congestion game.
\newblock In \emph{2010 First IEEE International Conference on Smart Grid
  Communications}, pages 495--500. IEEE, 2010.

\bibitem[Jin et~al.(2018)Jin, Allen-Zhu, Bubeck, and Jordan]{jin2018q}
Chi Jin, Zeyuan Allen-Zhu, Sebastien Bubeck, and Michael~I Jordan.
\newblock Is q-learning provably efficient?
\newblock \emph{Advances in neural information processing systems}, 31, 2018.

\bibitem[Jin et~al.(2021{\natexlab{a}})Jin, Liu, Wang, and Yu]{jin2021v}
Chi Jin, Qinghua Liu, Yuanhao Wang, and Tiancheng Yu.
\newblock V-learning--a simple, efficient, decentralized algorithm for
  multiagent rl.
\newblock \emph{arXiv preprint arXiv:2110.14555}, 2021{\natexlab{a}}.

\bibitem[Jin et~al.(2021{\natexlab{b}})Jin, Liu, Wang, and
  Yu]{jin2021vlearning}
Chi Jin, Qinghua Liu, Yuanhao Wang, and Tiancheng Yu.
\newblock V-learning -- a simple, efficient, decentralized algorithm for
  multiagent rl, 2021{\natexlab{b}}.

\bibitem[Johari and Tsitsiklis(2004)]{johari2004efficiency}
Ramesh Johari and John~N Tsitsiklis.
\newblock Efficiency loss in a network resource allocation game.
\newblock \emph{Mathematics of Operations Research}, 29\penalty0 (3):\penalty0
  407--435, 2004.

\bibitem[Kleinberg et~al.(2009)Kleinberg, Piliouras, and
  Tardos]{kleinberg2009multiplicative}
Robert Kleinberg, Georgios Piliouras, and {\'E}va Tardos.
\newblock Multiplicative updates outperform generic no-regret learning in
  congestion games.
\newblock In \emph{Proceedings of the forty-first annual ACM symposium on
  Theory of computing}, pages 533--542, 2009.

\bibitem[Krichene et~al.(2014)Krichene, Drigh{\`e}s, and
  Bayen]{krichene2014convergence}
Walid Krichene, Benjamin Drigh{\`e}s, and Alexandre Bayen.
\newblock On the convergence of no-regret learning in selfish routing.
\newblock In \emph{International Conference on Machine Learning}, pages
  163--171. PMLR, 2014.

\bibitem[Krichene et~al.(2015)Krichene, Drigh{\`e}s, and
  Bayen]{krichene2015online}
Walid Krichene, Benjamin Drigh{\`e}s, and Alexandre~M Bayen.
\newblock Online learning of nash equilibria in congestion games.
\newblock \emph{SIAM Journal on Control and Optimization}, 53\penalty0
  (2):\penalty0 1056--1081, 2015.

\bibitem[Lattimore and Szepesv{\'a}ri(2020)]{lattimore2020bandit}
Tor Lattimore and Csaba Szepesv{\'a}ri.
\newblock \emph{Bandit algorithms}.
\newblock Cambridge University Press, 2020.

\bibitem[Leonardos et~al.(2021)Leonardos, Overman, Panageas, and
  Piliouras]{leonardos2021global}
Stefanos Leonardos, Will Overman, Ioannis Panageas, and Georgios Piliouras.
\newblock Global convergence of multi-agent policy gradient in markov potential
  games, 2021.

\bibitem[Leslie(2004)]{leslie2004reinforcement}
David~S Leslie.
\newblock \emph{Reinforcement learning in games}.
\newblock PhD thesis, University of Bristol, 2004.

\bibitem[Leslie and Collins(2005)]{leslie2005individual}
David~S Leslie and Edmund~J Collins.
\newblock Individual q-learning in normal form games.
\newblock \emph{SIAM Journal on Control and Optimization}, 44\penalty0
  (2):\penalty0 495--514, 2005.

\bibitem[Leslie and Collins(2006)]{leslie2006generalised}
David~S Leslie and Edmund~J Collins.
\newblock Generalised weakened fictitious play.
\newblock \emph{Games and Economic Behavior}, 56\penalty0 (2):\penalty0
  285--298, 2006.

\bibitem[Liu et~al.(2021)Liu, Yu, Bai, and Jin]{liu2021sharp}
Qinghua Liu, Tiancheng Yu, Yu~Bai, and Chi Jin.
\newblock A sharp analysis of model-based reinforcement learning with
  self-play.
\newblock In \emph{International Conference on Machine Learning}, pages
  7001--7010. PMLR, 2021.

\bibitem[Macua et~al.(2018)Macua, Zazo, and Zazo]{macua2018learning}
Sergio~Valcarcel Macua, Javier Zazo, and Santiago Zazo.
\newblock Learning parametric closed-loop policies for markov potential games.
\newblock \emph{arXiv preprint arXiv:1802.00899}, 2018.

\bibitem[Marden(2012)]{marden2012state}
Jason~R Marden.
\newblock State based potential games.
\newblock \emph{Automatica}, 48\penalty0 (12):\penalty0 3075--3088, 2012.

\bibitem[Monderer and Shapley(1996)]{monderer1996potential}
Dov Monderer and Lloyd~S Shapley.
\newblock Potential games.
\newblock \emph{Games and economic behavior}, 14\penalty0 (1):\penalty0
  124--143, 1996.

\bibitem[Nikaid{\^o} and Isoda(1955)]{nikaido1955note}
Hukukane Nikaid{\^o} and Kazuo Isoda.
\newblock Note on non-cooperative convex games.
\newblock \emph{Pacific Journal of Mathematics}, 5\penalty0 (S1):\penalty0
  807--815, 1955.

\bibitem[Orabona(2019)]{orabona2019modern}
Francesco Orabona.
\newblock A modern introduction to online learning.
\newblock \emph{arXiv preprint arXiv:1912.13213}, 2019.

\bibitem[Osband and Van~Roy(2014)]{osband2014near}
Ian Osband and Benjamin Van~Roy.
\newblock Near-optimal reinforcement learning in factored mdps.
\newblock In Z.~Ghahramani, M.~Welling, C.~Cortes, N.~Lawrence, and K.Q.
  Weinberger, editors, \emph{Advances in Neural Information Processing
  Systems}, volume~27. Curran Associates, Inc., 2014.

\bibitem[Raghunathan et~al.(2019)Raghunathan, Cherian, and
  Jha]{raghunathan2019game}
Arvind Raghunathan, Anoop Cherian, and Devesh Jha.
\newblock Game theoretic optimization via gradient-based nikaido-isoda
  function.
\newblock In \emph{International Conference on Machine Learning}, pages
  5291--5300. PMLR, 2019.

\bibitem[Rosenberg and Mansour(2021)]{rosenberg2021oracle}
Aviv Rosenberg and Yishay Mansour.
\newblock Oracle-efficient regret minimization in factored mdps with unknown
  structure.
\newblock \emph{Advances in Neural Information Processing Systems}, 34, 2021.

\bibitem[Rosenthal(1973)]{rosenthal1973class}
Robert~W Rosenthal.
\newblock A class of games possessing pure-strategy nash equilibria.
\newblock \emph{International Journal of Game Theory}, 2\penalty0 (1):\penalty0
  65--67, 1973.

\bibitem[Roughgarden(2010)]{roughgarden2010algorithmic}
Tim Roughgarden.
\newblock Algorithmic game theory.
\newblock \emph{Communications of the ACM}, 53\penalty0 (7):\penalty0 78--86,
  2010.

\bibitem[Roughgarden and Tardos(2004)]{roughgarden2004bounding}
Tim Roughgarden and {\'E}va Tardos.
\newblock Bounding the inefficiency of equilibria in nonatomic congestion
  games.
\newblock \emph{Games and economic behavior}, 47\penalty0 (2):\penalty0
  389--403, 2004.

\bibitem[Rubinstein(2016)]{rubinstein2016settling}
Aviad Rubinstein.
\newblock Settling the complexity of computing approximate two-player nash
  equilibria.
\newblock In \emph{2016 IEEE 57th Annual Symposium on Foundations of Computer
  Science (FOCS)}, pages 258--265. IEEE, 2016.

\bibitem[Shapley(1953)]{shapley1953stochastic}
Lloyd~S Shapley.
\newblock Stochastic games.
\newblock \emph{Proceedings of the national academy of sciences}, 39\penalty0
  (10):\penalty0 1095--1100, 1953.

\bibitem[Song et~al.(2021)Song, Mei, and Bai]{song2021can}
Ziang Song, Song Mei, and Yu~Bai.
\newblock When can we learn general-sum markov games with a large number of
  players sample-efficiently?
\newblock \emph{arXiv preprint arXiv:2110.04184}, 2021.

\bibitem[Swenson et~al.(2018)Swenson, Murray, and Kar]{swenson2018best}
Brian Swenson, Ryan Murray, and Soummya Kar.
\newblock On best-response dynamics in potential games.
\newblock \emph{SIAM Journal on Control and Optimization}, 56\penalty0
  (4):\penalty0 2734--2767, 2018.

\bibitem[Szita and L{\H{o}}rincz(2009)]{szita2009optimistic}
Istv{\'a}n Szita and Andr{\'a}s L{\H{o}}rincz.
\newblock Optimistic initialization and greediness lead to polynomial time
  learning in factored mdps.
\newblock In \emph{Proceedings of the 26th annual international conference on
  machine learning}, pages 1001--1008, 2009.

\bibitem[Tian et~al.(2020)Tian, Qian, and Sra]{tian2020towards}
Yi~Tian, Jian Qian, and Suvrit Sra.
\newblock Towards minimax optimal reinforcement learning in factored markov
  decision processes.
\newblock \emph{Advances in Neural Information Processing Systems},
  33:\penalty0 19896--19907, 2020.

\bibitem[Xu and Tewari(2020)]{xu2020reinforcement}
Ziping Xu and Ambuj Tewari.
\newblock Reinforcement learning in factored mdps: Oracle-efficient algorithms
  and tighter regret bounds for the non-episodic setting.
\newblock \emph{Advances in Neural Information Processing Systems},
  33:\penalty0 18226--18236, 2020.

\bibitem[Zhang et~al.(2021{\natexlab{a}})Zhang, Yang, and
  Basar]{zhang2021multi}
Kaiqing Zhang, Zhuoran Yang, and Tamer Basar.
\newblock Multi-agent reinforcement learning: A selective overview of theories
  and algorithms.
\newblock \emph{Handbook of Reinforcement Learning and Control}, pages
  321--384, 2021{\natexlab{a}}.

\bibitem[Zhang et~al.(2020)Zhang, Shen, Mokhtari, Hassani, and
  Karbasi]{zhang2020one}
Mingrui Zhang, Zebang Shen, Aryan Mokhtari, Hamed Hassani, and Amin Karbasi.
\newblock One sample stochastic frank-wolfe.
\newblock In \emph{International Conference on Artificial Intelligence and
  Statistics}, pages 4012--4023. PMLR, 2020.

\bibitem[Zhang et~al.(2021{\natexlab{b}})Zhang, Ren, and Li]{zhang2021gradient}
Runyu Zhang, Zhaolin Ren, and Na~Li.
\newblock Gradient play in stochastic games: stationary points, convergence,
  and sample complexity.
\newblock \emph{arXiv preprint arXiv:2106.00198}, 2021{\natexlab{b}}.

\end{thebibliography}
\bibliographystyle{plainnat}


\newpage
\appendix
\tableofcontents

\section{Additional Motivating Examples}
\label{sec:additional_examples}
In this section, we present two additional motivating examples of our proposed models.
\begin{example}[{\bf Web Advertisements}]
Consider a set of websites as the facility set and companies who want to advertise their products as the players. Due to budget constraints, each company may only choose some of these websites to put its product ad. For each website, the probability that a user will click on a certain ad (and then buy the product) depends on how many ads are put on the website. If a website receives too many ads, the probability that a user can see a certain ad will decrease, thus making it congested.\footnote{Although the website's intelligent recommendation system may more or less mitigate this effect, it can be considered as a part of the reward function's property.} 
The reward each company will receive is measured by the amount of products sold during certain period of time, which is bandit feedback. 
\end{example}


\begin{example}[{\bf Server Usage}]
Consider a set of servers in a company as the facility set and server users as the players. Each user needs to request several servers to finish her computation task and the cost triggered from each server depends on the number of users requesting that server. 
Each user will try to minimize the total cost incurred from the servers she requested. As each user can see the cost from all the servers she requested, this is semi-bandit feedback. 
\end{example}

\section{Compute $\epsilon$-approximate Nash Equilibrium in Potential Games}
\label{sec:potential_game}

In this section, we show that the $\epsilon\textsc{-Nash}(\cdot)$ operation in Algorithm \ref{algo:nash_vi_matrix_game} can be computed efficiently by using Algorithm \ref{algo:eps NE}.

In particular, we first show that the matrix game with reward functions $\oq^k_1(\cdot), \dots, \oq^k_m(\cdot)$ used in Algorithm \ref{algo:nash_vi_matrix_game} is a potential game in Lemma \ref{lmm:q_potential_game}. Then, we show that Algorithm \ref{algo:eps NE} can efficiently compute an $\epsilon$-approximate Nash equilibrium for potential games and output a product policy as shown in Lemma \ref{lmm:compute_epsilon_nash}.

\begin{lemma}
    \label{lmm:q_potential_game}
    In line \ref{line:compute_nash_mg} of Algorithm \ref{algo:nash_vi_matrix_game}, the matrix game with reward functions $\oq^k_1(\cdot), \dots, \oq^k_m(\cdot)$ forms a potential game for both settings of semi-bandit feedback and bandit feedback.
\end{lemma}
\begin{proof}
    In the setting of semi-bandit feedback, since $\oq^k_i(\a)=\sum_{f\in a_i}(\hat{r}^{k, f}+b^{k, f, \reward})(\a)$, the reward functions $\oq^k_1(\cdot), \dots, \oq^k_m(\cdot)$ form a congestion game, which we know is a potential game \citep{monderer1996potential}.

    In the setting of bandit feedback, notice that by defining $\tilde{r}^{k, f}(i)=\widehat{\theta}^k_{i+m(f-1)}$ for $(i, f)\in[m]\times\F$, we can have $\tilde{r}^k_i(\a)=\inner{A_i(\a), \widehat{\theta}^k}=\sum_{f\in a_i}\tilde{r}^{k, f}(n^f(\a))$. Therefore, we claim that the desired potential function is
    $$\Phi^k(\a)=\widetilde{\Phi}^k(\a)+\tilde{b}^{k, \reward}(\a),\quad\text{where}\quad \widetilde{\Phi}^k(\a)=\sum_{f\in\mc{F}}\sum_{i=1}^{n^f(\a)}\tilde{r}^{k, f}(i).$$ 
    To see this, by referring to the definition of potential function in congestion game \citep{monderer1996potential}, since $\tilde{r}^k_i(\a)=\sum_{f\in a_i}\tilde{r}^{k, f}(n^f(\a))$, we have that
    $$\widetilde{\Phi}^k(a_i, a_{-i})-\widetilde{\Phi}^k(a_i', a_{-i})=\tilde{r}_i(a_i, a_{-i})-\tilde{r}_i(a_i', a_{-i}).$$
    As a result, we have
    \begin{align*}
        &\Phi^k(a_i, a_{-i})-\Phi^k(a_i', a_{-i})\\
        =&\Sp{\tilde{r}_i(a_i, a_{-i})+\tilde{b}^{k, \reward}(a_i, a_{-i})}-\Sp{\tilde{r}_i(a_i', a_{-i})+\tilde{b}^{k, \reward}(a_i', a_{-i})}\\
        =&\oq^k_i(a_i, a_{-i})-\oq^k_i(a_i', a_{-i}),
    \end{align*}
    which means that $\oq^k_1(\cdot), \dots, \oq^k_m(\cdot)$ form a potential game.
\end{proof}



\begin{lemma}
\label{lmm:compute_epsilon_nash}
Algorithm \ref{algo:eps NE} can output an $\epsilon$-approximate Nash equilibrium. 
\end{lemma}

\begin{proof}
Note that if at round $k$, we have $\max_{i\in[m]}\Delta_i\leq\epsilon$, then $\pi^k$ is an $\epsilon$-approximate Nash equilibrium. So we only need to prove that $\max_{i\in[m]}\Delta_i\leq\epsilon$ is satisfied at some round $k\in\{1, \dots, \left\lceil\frac{mr_\mathrm{max}}{\epsilon}\right\rceil\}$. 

Suppose the potential game $\Sp{\Bp{\A_i}_{i=1}^{m}, \Bp{r_i}_{i=1}^{m}}$ is associated with potential function $\Phi\in[0,\Phi_\mathrm{max}]$. Set $\pi^*=\argmax_{\pi\in\prod_{i\in[m]}\Delta(\A_i)}\Phi(\pi)$. Then for any $\pi\in\prod_{i\in[m]}\Delta(\A_i)$, we have
\begin{align*}
    \Phi(\pi^*)-\Phi(\pi)=&\sum_{i\in[m]}\Sp{\Phi(\pi^*_{1:i},\pi_{i+1:m})-\Phi(\pi^*_{1:i-1},\pi_{i:m})}\\
    =&\sum_{i\in[m]}\Sp{V_i^{\pi^*_{1:i},\pi_{i+1:m}}-V_i^{\pi^*_{1:i-1},\pi_{i:m}}}\\
    \leq& mr_\mathrm{max}. 
\end{align*}
As a result, we can set $\Phi_\mathrm{max}=mr_\mathrm{max}$. On the other hand, if $j=\argmax_{i\in[m]}\Delta_i$ for round $k$, we have
\begin{align*}
    \Phi(\pi^{k+1})-\Phi(\pi^k)=&\Phi(\pi_j^{k+1},\pi^k_{-j})-\Phi(\pi^k)\\
    =&V_j^{\pi_j^{k+1},\pi^k_{-j}}-V_j^{\pi^k}\\
    =&r_j(a_j^{k+1},\pi^k_{-j})-r_j(\pi^k)\tag{$\pi^k$ is deterministic}\\
    =&\Delta_j\\
    =&\max_{i\in[m]}\Delta_i. 
\end{align*}
So there must exist $k\in\{1, \dots, \left\lceil\frac{mr_\mathrm{max}}{\epsilon}\right\rceil\}$ such that $\max_{i\in[m]}\Delta_i\leq\epsilon$, otherwise $\Phi(\pi^k)$ increase at least $\epsilon$ at each round, which contradicts $\Phi\in[0,mr_\mathrm{max}]$. 
\end{proof}


\section{Analysis for Algorithm \ref{algo:nash_vi_matrix_game}}
\label{sec:ne_proof_mg}
Recall that the update rule in Algorithm \ref{algo:nash_vi_matrix_game} is $\oq^k_i(\a)=\hat{r}^k_i(\a)+b_i^{k, \reward}(\a)$, where we have
$$b_i^{k, \reward}(\a)=\sum_{f\in a_i}b^{k, f, \reward}(\a),\quad\text{and}\quad b^{k, f, \reward}(\a)=\sqrt{\frac{\tilde{\iota}}{N^{k, f}(n^f(\a))\vee 1}}.$$
For proof convenience, we define auxiliary value functions
\begin{align*}
    &\underline{Q}^k_{i}(\bm{a})=\hat{r}^k_{i}(\a)-b_i^{k, \reward}(\bm{a}),\\
    &\ov^k_i=\E_{\a\sim\pi^k}[\oq^k_i(\a)]\quad\text{and}\quad\uv^k_i=\E_{\a\sim\pi^k}[\uq^k_i(\a)].
\end{align*}
With these definitions, we now begin to prove Theorem \ref{theo:ne_mg}.

\begin{proof}[Proof of Theorem \ref{theo:ne_mg}]
    \textbf{Semi-bandit Feedback.} By the update rules in Algorithm \ref{algo:nash_vi_matrix_game}, in the setting of semi-bandit feedback, with probability at least $1-\delta$, simultaneously for all $(k, i, \a)\in[K]\times[m]\times\A$, we have
    $$\oq^k_i(\a)-r_i(\a)=\sum_{f\in a_i} \Mp{(\hat{r}^{k, f} - r^f)(\a) + b^{k, f, \reward}(\a)}\geq 0.$$
    The second inequality above is obtained by using standard Hoeffding's inequality and union bound, Therefore, we have $\oq^k_i(\a)\geq r_i(\a)$.
    
    Then, since $\pi^k$ is the $\epsilon$-approximate Nash equilibrium policy of $\oq^k_1, \dots, \oq^k_m$, we have
    \begin{align*}
      \ov^k_i=\E_{\a\sim\pi^k}[\oq^k_i(\a)]=&\max_{\nu\in\Delta(\A_i)}\E_{\a\sim(\nu, \pi^k_{-i})}[\oq^k_i(\a)]-\epsilon\\
      \geq& \max_{\nu\in\Delta(\A_i)}\E_{\a\sim(\nu, \pi^k_{-i})}[r_i(\a)]-\epsilon=V^{\dagger, \pi^{k}_{-i}}_i-\epsilon.  
    \end{align*}
    Meanwhile, by definition of $\uq^k_i(\a)$ and $\uv^k_i$, we can similarly show that $\uq^k_i(\a)\leq r_i(\a)$ and $\uv^k_i\leq V^{\pi^k}_i$. Therefore, we can have $V^{\dagger, \pikni}_i-V^{\pi^k}_i\leq\ov^k_i-\uv^k_i+\epsilon$.
    
     Now, we define $\tq^k(\a)=\max_{i\in[m]}2b_i^{k, \reward}(\a)$ and $\tv^k=\E_{\a\sim\pi^k}[\tq^k(\a)]$. Then, we can notice that
    $$\max_{i\in[m]}(\oq^k_i-\uq^k_i)(\a)\leq\max_{i\in[m]}2b_i^{k, \reward}(\a)=\tq^k(\a),$$
    $$\max_{i\in[m]}(\ov^k_i-\uv^k_i)\leq\E_{\a\sim\pi^k}\Mp{\max_{i\in[m]}(\oq^k_i-\uq^k_i)(\a)}\leq \E_{\a\sim\pi^k}[\tq^k(\a)]=\tv^k.$$
    
    We further define 
    $\M^k=\E_{\a\sim\pi^k}\Mp{\tq^k(\a)}-\tq^k(\a^k)=\tv^k-\tq^k(\a^k)$.
    It is not hard to verify that $\M^k$ is a martingale difference sequence with respect to the history from episode $1$ to $k-1$. 
    Meanwhile, since $\abs{b^{k, \reward}(\a)}=\sum_{f\in\F}\sqrt{\frac{\tilde{\iota}}{N^{k, f}(n^f(\a))\vee 1}}\leq F\sqrt{\tilde{\iota}}$. 
    Thus, by Azuma-Hoeffding inequality, we have $\sum_{k=1}^K\M^k=\tmco\Sp{F\sqrt{K}}$. Therefore, we have
    \begin{align*}
        \text{Nash-Regret}(K)=&\sum_{k=1}^{K}\max_{i\in[m]}\Sp{V^{\dagger, \pikni}_i-V^{\pi^k}_i}\\
        = & \sum_{k=1}^{K}\min\Bp{\max_{i\in[m]}\Sp{V^{\dagger, \pikni}_i-V^{\pi^k}_i}, F}\tag{Since the value is always bounded by $F$.}\\
        \leq & \sum_{k=1}^K\min\Bp{\max_{i\in[m]}\Sp{\ov^k_i-\uv^k_i}, F}+K\epsilon\\
        \leq & \sum_{k=1}^K\min\Bp{\tv^k, F}+K\epsilon\\
        = & \sum_{k=1}^K\Sp{\min\Bp{\tq^k(\a^k), F}+\M^k}+K\epsilon\\
        \leq & \tmco\Sp{F\sqrt{K}} + 2\sum_{k=1}^K\Bp{\max_{i\in[m]}b_i^{k, \reward}(\a^k), F}\tag{By taking $\epsilon=1/K$.}\\
        \leq & \tmco\Sp{F\sqrt{K}} + 2\sum_{f\in\F}\sum_{k=1}^K\sqrt{\frac{\tilde{\iota}}{N^{k, f}(n^f(\a^k))\vee 1}}\\
        \leq & \tmco\Sp{F\sqrt{mK}}\tag{By Lemma \ref{lmm:sum_sqrtn_mg}.}
    \end{align*}
    
    \textbf{Bandit Feedback.} By using Lemma \ref{lmm:least_square_bound_mg}, which guarantees optimistic estimation, we can similarly show that
    $$\text{Nash-Regret}(K)\leq\sum_{k=1}^K\M^k+\sum_{k=1}^K\min\Bp{2\tilde{b}^{k, \reward}(\a^k), F}+K\epsilon.$$
    To have an upper bound on $\M^k$ here, recall that $\tilde{b}^{k, \reward}(\a)=\max_{i\in[m]}\Norm{A_i(\a)}_{\Sp{V^k}^{-1}}\sqrt{\tilde{\beta}_k}$ and $\sqrt{\tilde{\beta}_K}=\tmco\Sp{\sqrt{F\tilde{d}}}=\tmco\Sp{F\sqrt{m}}$. Meanwhile, we have $\Norm{A_i(\a)}_{\Sp{V^k}^{-1}}\leq\Norm{A_i(\a)}_{I}=\Norm{A_i(\a)}_2\leq \sqrt{F}$. Thus, we have $\abs{\M^k}\leq\tmco\Sp{\sqrt{mF^3}}$, which by Azuma-Hoeffding inequality implies $\sum_{k=1}^K\M^k=\tmco\Sp{\sqrt{mF^3K}}$.
    
    Then the sum of the bonus terms can be bounded by using Lemma \ref{lmm:elliptical_potential_variant}. In particular, with $\epsilon=1/K$, we have
    \begin{align*}
        \text{Nash-Regret}(K)\leq&\tmco\Sp{\sqrt{mF^3K}} + 2\sum_{k=1}^K\min\Bp{\max_{i\in[m]}\Norm{A_i(\a^k)}_{\Sp{V^k}^{-1}}\sqrt{\tilde{\beta}_k}, F}\\
        \leq & \tmco\Sp{\sqrt{mF^3K}} + 2\sqrt{K\sum_{k=1}^{K}\min\Bp{\max_{i\in[m]}\Norm{A_i(\a^k)}_{\Sp{V^k}^{-1}}^2\tilde{\beta}_k, F^2}}\\
        \leq & \tmco\Sp{\sqrt{mF^3K}} + \sqrt{\tmco\Sp{mF^2K}\sum_{k=1}^K\min\Bp{\max_{i\in[m]}\Norm{A_i(\a^k)}_{\Sp{V^k}^{-1}}^2, 1}}\tag{Since $\tilde{\beta}_k=\tmco\Sp{mF^2}$.}\\
        \leq &\tmco\Sp{\sqrt{mF^3K}} + \tmco\Sp{\sqrt{mF^2K\cdot mF}}\tag{By Lemma \ref{lmm:elliptical_potential_variant}.}\\
        \leq &\tmco\Sp{mF^{3/2}\sqrt{K}}.
    \end{align*}
\end{proof}

\subsection{Lemmas for Bandit Feedback}

The following lemma, as a direct corollary of the confidence bound for least square estimators, shows that the reward estimation error can be bounded by the reward bonus term.
\begin{lemma}
	\label{lmm:least_square_bound_mg}
	With probability at least $1-\delta$, simultaneously for all $(i, k, \a)$, it holds that $|(\tilde{r}^k_{i}-r_{i})(\a)|\leq\tilde{b}^{k, \reward}(\a)$, where $\tilde{r}^k_{i}$ and $\tilde{b}^{k, \reward}$ are defined in \eqref{equ:r_estimate_bandit}.
\end{lemma}
\begin{proof}
	By construction, we have
	\begin{align*}
		|(\tilde{r}^k_{i}-r_{i})(\a)|=&\abs{\inner{A_i(\a), \widehat{\theta}-\theta}}\\
		\leq & \Norm{A_i(\a)}_{\Sp{V^k}^{-1}}\Norm{\widehat{\theta}-\theta}_{V^k}\\
		\overset{\text{(i)}}{\leq} & \Norm{A_i(\a)}_{\Sp{V^k}^{-1}}\Sp{\Norm{\theta}_2+\sqrt{F\log\Sp{\det(V^k)}+F\tilde{\iota}}},
	\end{align*}
	where the inequality (i) above holds because of Theorem 20.5 in \cite{lattimore2020bandit} and the fact that the reward noise is $\sqrt{F}$-subGaussian. Since each element in $\theta$ is bounded in $[0, 1]$ by construction, we have $\Norm{\theta}_2\leq\sqrt{\tilde{d}}$. 
	
	Then, by Lemma \ref{lmm:elliptical_potential_variant}, we have $\det\Sp{V^k}\leq \Sp{1+\frac{mkF}{\tilde{d}}}^{\tilde{d}}$ since by construction $\Norm{A_i(\a)}_2^2\leq F$.
	
	Finally, to make this bound valid for all player $i\in[m]$, we only need to take maximization over $i\in[m]$. Therefore, with probability at least $1-\delta$, we have
	$$|(\tilde{r}^k_{i}-r_{i})(\a)|\leq\max_{i\in[m]}\Norm{A_i(\a)}_{\Sp{V^k}^{-1}}\sqrt{\tilde{\beta}_k}=\tilde{b}^{k, \reward}(\a),$$
	where $\sqrt{\tilde{\beta}_k}=\sqrt{\tilde{d}}+\sqrt{F\tilde{d}\log\Sp{1+\frac{mkF}{\tilde{d}}}+F\tilde{\iota}}$.
\end{proof}



The following is a variant of the famous elliptical potential lemma, which helps bound the sum of reward bonus under bandit feedback. Here, we apply some techniques from the proof of Lemma 19.4 in \cite{lattimore2020bandit}.

\begin{lemma}
    \label{lmm:elliptical_potential_variant}
    Let $K, m\geq 1$ be integers. Suppose $V^k=I+\sum_{k'=1}^{k-1}\sum_{i=1}^{m}A_i^{k'}\Sp{A_i^{k'}}^\top$, where $A_i^{k'}\in\R^d$ and $\Norm{A_i^{k'}}_2^2\leq F$. Then, it holds that
    $$\det\Sp{V^k}\leq\Sp{1+\frac{mkF}{d}}^d,\quad\text{and}\quad\sum_{k=1}^K\min\Bp{\max_{i\in[m]}\Norm{A_i^k}^2_{\Sp{V^k}^{-1}}, 1}\leq 2d\log\Sp{1+\frac{mKF}{d}}.$$
\end{lemma}
\begin{proof}
    For the first upper bound about $\det\Sp{V^k}$, we have
    \begin{align*}
		\det\Sp{V^k} =&\prod_{j=1}^{d}\lambda_j\tag{$\lambda_1, \dots, \lambda_d$ are eigenvalues of $V^k$}\\
		\leq & \Sp{\frac{\mathrm{tr}\Sp{V^k}}{d}}^d\tag{By AM-GM inequality}\\
		= & \Sp{\frac{\mathrm{tr}\Sp{I}+\sum_{k'=1}^{k-1}\sum_{i=1}^{m}\Norm{A_i^{k'}}_2^2}{d}}^d\\
		\leq & \Sp{1+\frac{mkF}{d}}^d.\tag{Since $\Norm{A_i^{k'}}_2^2\leq F$.}
	\end{align*}
	
	For the second upper bound. First, we notice that $\min\Bp{1, x}\leq 2\log(1+x)$ for any $x\geq 0$. Thus, we have
	$$\sum_{k=1}^K\min\Bp{1, \max_{i\in[m]}\Norm{A_i^k}^2_{\Sp{V^k}^{-1}}}\leq 2\sum_{k=1}^K\log\Sp{1+\max_{i\in[m]}\Norm{A_i^k}^2_{\Sp{V^k}^{-1}}}.$$
	Then, for $k\geq 2$, we can notice that
	\begin{align*}
		V^k=&V^{k-1}+\sum_{i=1}^{m}A_i^{k-1} \Sp{A_i^{k-1}}^\top\\
		=&\Sp{V^{k-1}}^{1/2}\Sp{I+\Sp{V^{k-1}}^{-1/2}\Sp{\sum_{i=1}^{m}A_i^{k-1} \Sp{A_i^{k-1}}^\top}\Sp{V^{k-1}}^{-1/2}}\Sp{V^{k-1}}^{1/2}\\
		=&\Sp{V^{k-1}}^{1/2}\Sp{I+\sum_{i=1}^{m}\Sp{\Sp{V^{k-1}}^{-1/2}A_i^{k-1}}\Sp{\Sp{V^{k-1}}^{-1/2}A_i^{k-1}}^\top}\Sp{V^{k-1}}^{1/2}.
	\end{align*}
	Therefore, we have
	\begin{align*}
		\det\Sp{V^k}=&\det\Sp{V^{k-1}}\det\Sp{I+\sum_{i=1}^{m}\Sp{\Sp{V^{k-1}}^{-1/2}A_{i}^{k-1}}\Sp{\Sp{V^{k-1}}^{-1/2}A_i^{k-1}}^\top}\\
		\geq & \det\Sp{V^{k-1}}\Sp{1+\max_{i\in[m]}\Norm{A_i^{k-1}}^2_{\Sp{V^{k-1}}^{-1}}}\tag{By Lemma \ref{lmm:det_ysum}.}\\
		\geq & \prod_{k'=1}^{k-1}\Sp{1+\max_{i\in[m]}\Norm{A_{ i}^{k'}}^2_{\Sp{V^{k'}}^{-1}}}.\tag{Since by definition, $V^1=I$.}
	\end{align*}
	As a result, we have
	\begin{align*}
	    \sum_{k=1}^K\min\Bp{\max_{i\in[m]}\Norm{A_i^k}^2_{\Sp{V^k}^{-1}}, 1}\leq & 2\sum_{k=1}^K\log\Sp{1+\max_{i\in[m]}\Norm{A_i^k}^2_{\Sp{V^k}^{-1}}}\\
	    & \leq 2\log\Sp{\det\Sp{V^{K+1}}}\\
	    & \leq 2d\log\Sp{1+\frac{mKF}{d}}.
	\end{align*}
\end{proof}

\subsection{Technical Lemmas}

\begin{lemma}
	\label{lmm:det_ysum}
	Let $y_1, \dots, y_m\in\R^d$ be a set of vectors. Then, it holds that
	$$\det\Sp{I+\sum_{i=1}^{m}y_iy_i^\top}\geq 1+\max_{i\in[m]}\Norm{y_i}_2^2.$$
\end{lemma}
\begin{proof}
	Since $I+\sum_{i=1}^{m}y_iy_i^\top\succeq I+y_iy_i^\top$ for any $i\in[m]$, we have $\det\Sp{I+\sum_{i=1}^{m}y_iy_i^\top}\geq\det\Sp{I+y_iy_i^\top}$ for any $i\in[m]$. That is, we have
	$$\det\Sp{I+\sum_{i=1}^{m}y_iy_i^\top}\geq\max_{i\in[m]}\det\Sp{I+y_iy_i^\top}=1+\max_{i\in[m]}\Norm{y_i}_2^2.$$
	The last line above holds because the matrix $I+y_iy_i^\top$ has eigenvalues $1+\Norm{y_i}_2^2$ and 1.
\end{proof}

\begin{lemma}
	\label{lmm:sum_sqrtn_mg}
	For any $f\in\mc{F}$, it holds that
	$$\sum_{k=1}^{K}\sqrt{\frac{1}{\N^{k, f}(n^f(\a^k))\vee 1}}\leq \tmco\Sp{\sqrt{mK}}.$$
\end{lemma}
\begin{proof}
    Here, we have
	\begin{align*}
		\sum_{k=1}^{K}\sqrt{\frac{1}{N^{k, f}(n^f(\a^k))\vee 1}}=&\sum_{n=0}^{m}\sum_{\ell=1}^{N^{K, f}(n)}\sqrt{\frac{1}{\ell}}\\
		\leq & 2 \sum_{n=0}^{m}\sqrt{N^{K, f}(n)}\tag{By standard technique}\\
		\leq & 2\sqrt{(m+1)\sum_{n=0}^{m}N^{K, f}(n)}\\
		=&\tmco\Sp{\sqrt{mK}}.
	\end{align*}
	The last equality above is based on a pigeon-hold principle argument similar to Lemma \ref{lmm:sum_sqrtn}.
\end{proof}

\section{Analysis for Algorithm \ref{algo:FW}}

\subsection{Exploration Distribution and Smoothness}
\label{sec:explore_dist}

We choose the exploration distribution to be the G-optimal design and we have the following properties. 

\begin{lemma}\label{lemma:unbiased}
(Unbiasedness) For any episode $k\in[K]$, $i\in[m]$ and $a\in\A_i$, we have
$$\E_k\Mp{\widehat{\nabla}_i^k\Phi(a)}=\nabla_i^k\Phi(a),$$
where $\E_k[\cdot]$ is taken over all the randomness before episode $k$. 
\end{lemma}

\begin{proof} By the definition of $\widehat{\nabla}_i^k\Phi(a)$, we have
\begin{align*}
    \E_k\Mp{\widehat{\nabla}_i^k\Phi(a)}=&\E_k\inner{\phi_i(a),\widehat{\theta}_i^k(\pi^k)}\\
    =&\E_k\Mp{\frac{1}{\tau}\sum_{t=1}^\tau\phi_i(a)^\top[\Sigma_i^{k}]^{-1}\phi_i(a_i^{k,t})r_i^{k,t}}\\
    =&\E_k\Mp{\phi_i(a)^\top[\Sigma_i^{k}]^{-1}\phi_i(a_i^{k,1})r_i^{k,1}}\\
    =&\E_k\Mp{\phi_i(a)^\top[\Sigma_i^{k}]^{-1}\phi_i(a_i^{k,1})\phi_i(a_i^{k,1})^\top\theta_i^{k,1}(\pi^k)}\\
    =&\sum_{a_i^k\in\A_i}\pi_i^k(a_i^{k,1})\phi_i^\top(a)[\Sigma_i^k]^{-1}\phi_i(a_i^{k,1})\phi_i(a_i^{k,1})^\top\theta_i(\pi^k)\tag{$a_i^{k,1}$ only depends on $\pi_i^k$ and $\theta_i^{k,1}(\pi^k)$ only depends on $\pi_{-i}^k$}\\
    =&\phi_i^\top(a)[\Sigma_i^k]^{-1}\Mp{\sum_{a_i^k\in\A_i}\pi_i^k(a_i^{k,1})\phi_i(a_i^{k,1})\phi_i(a_i^{k,1})^\top}\theta_i(\pi^k)\\
    =&\phi_i^\top(a_i)\theta_i(\pi^k)\\
    =&\nabla_i^k\Phi(a).
\end{align*}
\end{proof}

\begin{lemma}\label{lemma:estimate bound}
For any episode $k\in[K]$, $i\in[m]$ and $a\in\A_i$, we have
$$\abs{\phi_i(a)^\top[\Sigma_i^{k}]^{-1}\phi_i(a_i^{k,t})r_i^{k,t}}\leq\frac{F^2}{\gamma}.$$
\end{lemma}

\begin{proof}
As $\pi_i^k=(1-\gamma)(\nu\widetilde{\pi}_i^k+(1-\gamma)\pi_i^{k-1})+\gamma\rho_i$, we have
\begin{align*}
    \Sigma_i^k=&\E_{a_i\sim\pi_i^{k}}\phi_i(a_i) \phi_i(a_i)^\top
    \succeq\gamma \E_{a_i\sim\rho_i}\phi_i(a_i) \phi_i(a_i)^\top,
\end{align*}
and $\rho_i$ is the G-optimal design with respect to $\phi_i(\cdot)$, for any action $a\in\A_i$ we have
$$\Norm{\phi_i(a)}_{[\Sigma_i^k]^{-1}}^2\leq \frac{1}{\gamma}\Norm{\phi_i(a)}_{[\E_{a_i\sim\rho_i}\phi_i(a_i) \phi_i(a_i)^\top]^{-1}}^2\leq \frac{F}{\gamma}.$$
Then for any $t\in[\tau]$, since $|r^{k, t}_i|\leq F$, we have
\begin{align*}
    \abs{r_i^{k,t}\phi_i^\top(a)[\Sigma_i^k]^{-1} \phi_i(a_i^{k,t})}\leq \abs{r_i^{k,t}}\Norm{\phi_i(a)}_{[\Sigma_i^k]^{-1}}\Norm{\phi_i(a_i^{k,t})}_{[\Sigma_i^k]^{-1}}\leq\frac{F^2}{\gamma}.
\end{align*}
As a result, we have
$$\abs{\widehat{\nabla}_i^k\Phi(a)}=\abs{\frac{1}{\tau}\sum_{t=1}^\tau\phi_i(a)^\top[\Sigma_i^{k}]^{-1}\phi_i(a_i^{k,t})r_i^{k,t}}\leq\frac{F^2}{\gamma}$$
\end{proof}

\begin{lemma}\label{lemma:variance bound}
For any episode $k\in[K]$, $i\in[m]$ and $a\in\A_i$, we have
$$\E_k\Mp{\Sp{\phi_i(a)^\top[\Sigma_i^{k}]^{-1}\phi_i(a_i^{k,t})r_i^{k,t}}^2}\leq\frac{F^3}{\gamma}.$$
\end{lemma}

\begin{proof}
We first show that for any $t\in[\tau]$, we have
\begin{align*}
    &\E_k\Mp{\Sp{\phi_i(a)^\top[\Sigma_i^{k}]^{-1}\phi_i(a_i^{k,t})r_i^{k,t}}^2}\\
    \leq& F^2\E_k\Mp{\Sp{\phi_i(a)^\top[\Sigma_i^{k}]^{-1}\phi_i(a_i^{k,t})}^2}\\
    \leq& F^2\E_k\Mp{\phi_i(a)^\top[\Sigma_i^{k}]^{-1}\phi_i(a_i^{k,t})\phi_i(a_i^{k,t})^\top[\Sigma_i^{k}]^{-1}\phi_i(a)^\top}\\
    =& F^2\phi_i(a)^\top[\Sigma_i^{k}]^{-1}\phi_i(a)\\
    \leq& \frac{F^3}{\gamma}. 
\end{align*}

\end{proof}

\begin{lemma}\label{lemma:concentration}
With probability $1-\delta$, for all $k\in[K]$, $i\in[m]$ and $a\in\A_i$, we have
$$\abs{\widehat{\nabla}_i^k\Phi(a)-\nabla_i^k\Phi(a)}\leq c\sqrt{\frac{F^4\log(mK/\delta)}{\gamma \tau}}+\frac{cF^3\log(mK/\delta)}{\gamma \tau}$$
\end{lemma}

\begin{proof}
Recall that
$$\widehat{\nabla}^k_i\Phi(a_i)=\frac{1}{\tau}\sum_{t=1}^\tau \phi_i^\top(a_i)[\Sigma_i^k]^{-1} r_i^{k,t}\phi_i(a_i^{k,t}),$$
and $(a_i^{k,t},r_i^{k,t})$ are drawn independently at each $t\in[\tau]$. Lemma \ref{lemma:unbiased} shows that $\widehat{\nabla}^k_i\Phi(a_i)$ is an unbiased estimate of $\nabla^k_i\Phi(a_i)$ In addition, Lemma \ref{lemma:estimate bound} shows that $\phi_i^\top(a_i)[\Sigma_i^k]^{-1} r_i^{k,t}\phi_i(a_i^{k,t})$ is bounded by $F^2/\gamma$ and Lemma \ref{lemma:variance bound} shows that its second moment is bounded by $F^3/\gamma$. Then by Bernstein's inequality, for a fixed $k\in[K]$, $i\in[m]$ and $a\in\A_i$, with probability $1-\delta$, we have
$$\abs{\widehat{\nabla}_i^k\Phi(a)-\nabla_i^k\Phi(a)}\leq \sqrt{\frac{2F^3\log(2/\delta)}{\gamma \tau}}+\frac{3F^2\log(2/\delta)}{2\gamma \tau}. $$
The argument holds by applying the union bound and the fact that $|\A_i|\leq 2^F$.

\end{proof}

\begin{lemma}
\label{lmm:smooth}
$\Phi(\cdot)$ is $mF$-Lipschitz and $mF$-smooth with respect to the L1 norm $\|\cdot\|_1$. 
\end{lemma}

\begin{proof}
Recall that $\Phi(\pi)=\E_{\a\sim\pi}\Phi(\a)$ and $\Phi(\a)\in[0,mF]$. 
\begin{align*}
    \Phi(\pi)-\Phi(\pi')=&\E_{\a\sim\pi}\Phi(\a)-\E_{\a\sim\pi'}\Phi(\a)\\
    =&\sum_{i\in[m]}\E_{a_{1:i-1}\sim\pi'_{1:i-1},a_{i:m}\sim\pi_{i:m}}\Phi(\a)-\E_{a_{1:i}\sim\pi'_{1:i},a_{i+1:m}\sim\pi_{i+1:m}}\Phi(\a)\\
    \leq&\sum_{i\in[m]}\Norm{\pi_i-\pi'_i}_1\cdot\Norm{\Phi}_\infty\\
    \leq& mF\Norm{\pi-\pi'}_1. 
\end{align*}
Similarly we have $\nabla_{\pi}\Phi(a_i)=\E_{a_{-i}\sim\pi_{-i}}\Phi(a_i,a_{-i})$. As a result, we have
$$\Norm{\nabla_{\pi}\Phi-\nabla_{\pi'}\Phi}_\infty\leq mF\Norm{\pi-\pi'}_1. $$
\end{proof}

\begin{definition}(Frank Wolfe Gap)
The Frank Wolfe gap of a joint strategy $\pi$ for $\Phi(\cdot)$ is defined as
$$G(\pi)=\max_{\pi'}\inner{\pi'-\pi,\nabla_\pi\Phi}.$$
\end{definition}

\begin{lemma}\label{lemma:gap to nash}
Suppose the Frank Wolfe gap of $\pi$ is $\epsilon$. Then $\pi$ is an $\epsilon$-Nash policy. 
\end{lemma}

\begin{proof}
For a fixed player $i$, suppose player $i$ change her strategy to $\pi'_i$. 
\begin{align*}
    V_i^{\pi'_i,\pi_{-i}}-V_i^{\pi}&=\Phi(\pi'_i,\pi_{-i})-\Phi(\pi)\\
    &=\inner{\pi'_i-\pi_i,\nabla_{\pi_i}\Phi}\\
    &\leq\max_{\pi'}\inner{\pi'-\pi,\nabla_\pi\Phi}\\
    &\leq\epsilon. 
\end{align*}
\end{proof}

\subsection{Analysis for Frank Wolfe in Bandit Feedback}
\label{sec:fw_proof}

\begin{theorem}
\label{theo:decentralized_bandit}
Let $T=K\tau$. For the congestion game with bandit feedback, by running Algorithm \ref{algo:FW} with gradient estimator $\widehat{\nabla}^k_i\Phi$ in \eqref{equ:grad_estimator} and exploration distribution $\rho_i$ in \eqref{equ:G_optimal}, setting parameters $\nu=\frac{F}{m\sqrt{K}}$, $\gamma=\frac{F}{mK}$ and $\tau=K^2$, if $K\geq\frac{2F}{m}$, then with probability $1-\delta$, we have
$$\text{Nash-Regret}(T)=\tau \sum_{k=1}^K G(\pi^k)=\tmco\Sp{m^2F^2T^{5/6}+m^3F^3T^{2/3}}.$$
\end{theorem}

\begin{proof}
Set $\nabla^k\Phi=\nabla\Phi(\Pi^k)\in\R^A$ and $\nabla_i^k\Phi=\nabla^k\Phi(\pi_i)\in\R^{A_i}$. As we have $\Phi(\cdot)$ is $mF$-smooth w.r.t. $\|\cdot\|_1$, we have
\begin{align*}
    \Phi(\pi^{k+1})\geq& \Phi(\pi^k)+\inner{\nabla \Phi(\pi^{k}),\pi^{k+1}-\pi^k}-\frac{mF}{2}\|\pi^{k+1}-\pi^k\|_1^2\\
    =& \Phi(\pi^k)+(1-\gamma)\nu\inner{\nabla \Phi(\pi^{k}),\widetilde{\pi}^{k+1}-\pi^k}+\gamma\inner{\nabla^k\Phi,\rho-\pi^k}\\
    &\qquad-\frac{mF}{2}(2\nu^2\Norm{\widetilde{\pi}^k-\pi^k}_1^2+2\gamma^2\Norm{\rho-\pi^k}_1^2)\\
    \geq& \Phi(\pi^k)+(1-\gamma)\nu\inner{\nabla \Phi(\pi^{k}),\widetilde{\pi}^{k+1}-\pi^k}-\gamma\Norm{\nabla^k\Phi}_\infty\Norm{\rho-\pi^k}_1\\
    &\qquad-\frac{mF}{2}(2\nu^2\Norm{\widetilde{\pi}^k-\pi^k}_1^2+2\gamma^2\Norm{\rho-\pi^k}_1^2)\\
    \geq& \Phi(\pi^k)+(1-\gamma)\nu\inner{\nabla \Phi(\pi^{k}),\widetilde{\pi}^{k+1}-\pi^k}-2\gamma m^2F-4m^3F(\nu^2+\gamma^2). \tag{By Lemma \ref{lmm:smooth}.}
\end{align*}
Define the true target policy at episode $k$
$$\widehat{\pi}_i^{k+1}=\argmax_{\pi_i}\inner{\pi_i,\nabla_i \Phi(\pi_i^k)},$$
and the Frank Wolfe gap of joint strategy $\pi$
$$G(\pi)=\max_{\pi'}\inner{\pi'-\pi,\nabla \Phi(\pi)}.$$
Then we have
\begin{align*}
    \inner{\nabla \Phi(\pi^{k}),\widetilde{\pi}^{k+1}-\pi^k}=&\inner{\widehat{\nabla}^k \Phi(\pi^{k}),\widetilde{\pi}^{k+1}-\pi^k}+\inner{\nabla \Phi(\pi^{k})-\widehat{\nabla}^k \Phi(\pi^{k}),\widetilde{\pi}^{k+1}-\pi^k}\\
    \geq& \inner{\widehat{\nabla}^k \Phi(\pi^{k}),\widehat{\pi}^{k+1}-\pi^k}+\inner{\nabla \Phi(\pi^{k})-\widehat{\nabla}^k \Phi(\pi^{k}),\widetilde{\pi}^{k+1}-\pi^k}\\
    =& \inner{\nabla\Phi(\pi^{k}),\widehat{\pi}^{k+1}-\pi^k}+\inner{\nabla \Phi(\pi^{k})-\widehat{\nabla}^k \Phi(\pi^{k}),\widetilde{\pi}^{k+1}-\widehat{\pi}^{k+1}}\\
    \geq& G(\pi^k)-2m\Norm{\nabla \Phi(\pi^{k})-\widehat{\nabla}^k \Phi(\pi^{k})}_\infty\\
    \geq& G(\pi^k)-c\sqrt{\frac{m^2F^4\log(mK/\delta)}{\gamma \tau}}-\frac{cmF^3\log(mK/\delta)}{\gamma \tau}
\end{align*}
Apply it to the previous bound and we have
\begin{align*}
    \Phi(\pi^{k+1})\geq&\Phi(\pi^k)+(1-\gamma)\nu G(\pi^k)-c\frac{(1-\gamma)\nu}{\sqrt{\gamma\tau}}\sqrt{m^2F^4\log(mK/\delta)}\\
    &\qquad-c\frac{(1-\gamma)\nu}{\gamma\tau}mF^3\log(mK/\delta)-\gamma 2m^2F-4m^3F(\nu^2+\gamma^2).
\end{align*}
Summing over $k\in[K]$ and we get
\begin{align*}
   \sum_{k=1}^K G(\pi^k)\leq&\frac{\Phi(\pi^{K+1})-\Phi(\pi^{1})}{(1-\gamma)\nu }+c\frac{K}{\sqrt{\gamma\tau}}\sqrt{m^2F^4\log(mK/\delta)}+c\frac{K}{\gamma\tau}mF^3\log(mK/\delta)\\
   &\qquad +\frac{2m^2FK\gamma}{(1-\gamma)\nu}+\frac{4(\nu^2+\gamma^2)m^3FK}{(1-\gamma)\nu}. 
\end{align*}
Set $\nu=\frac{F}{m\sqrt{K}}$, $\gamma=\frac{F}{mK}$, $\tau=K^2$ and notice that when $K\geq\frac{2F}{m}$, we have $1-\gamma\geq\frac{1}{2}$. Since $\Phi(\cdot)$ is bounded in $[0, mF]$, we can have
$$\sum_{k=1}^K G(\pi^k)=\tmco\Sp{m^2F^2K^{1/2}+m^3F^3}.$$
Then by Lemma \ref{lemma:gap to nash}, for $T=K\tau$, we have
$$\text{Nash-Regret}(T)=\tau \sum_{k=1}^K G(\pi^k)=\tmco\Sp{m^2F^2T^{5/6}+m^3F^3T^{2/3}}.$$
\end{proof}

\subsection{Algorithm and Analysis for Semi-bandit Feedback}
\label{sec:fw_semi_bandit}

In the setting of semi-bandit feedback, we will need a different gradient estimator $\widetilde{\nabla}^k_i\Phi(a_i)$ and a different exploration distribution $\Tilde{\rho}_i$ to utilize the extra reward information from each chosen facility.

Based on the analysis in Section \ref{sec:decentralized_algo}, using \eqref{equ:phi_grad}, we have $\nabla^k_i\Phi(a_i)=\sum_{f\in a_i}[\theta_i(\pi^k)]_f$, where $[\theta_i(\pi^k)]_f=\E_{a_{-i}\sim\pi^k_{-i}}\Mp{r^f(n^f(a_{-i})+1)}$. Meanwhile, in semi-bandit feedback, the mean of $t$-th reward player $i$ received for facility $f$ at episode $k$ is $r^f(n^f(a^{k, t}_i, a^{k, t}_{-i}))$. Therefore, we can use inverse propensity score (IPS) estimator to estimate $[\theta_i(\pi^k)]_f$. In particular, we have
$$[\widetilde{\theta}^k_i(\pi^k)]_f=\frac{1}{\tau}\sum_{t=1}^{\tau}[\widetilde{\theta}^{k, t}_i(\pi^k)]_f,\quad\text{where}\quad[\widetilde{\theta}^{k, t}_i(\pi^k)]_f=\frac{r^{k, t, f}\mathds{1}\Bp{f\in a^{k, t}_i}}{\P_{a_i\sim\pi_i^k}(f\in a_i)}.$$
Then, we can naturally have 
\begin{equation}
    \label{equ:grad_estimator_semi}
    \widetilde{\nabla}^k_i\Phi(a_i)=\sum_{f\in a_i}[\widetilde{\theta}^k_i(\pi^k)]_f.
\end{equation}
Furthermore, by Lemma \ref{lmm:tilde_rho}, we can see that by using $\tilde{\rho}_i$ computed by Algorithm \ref{algo:compute_tilde_rho}, for all players, we have $\P_{a_i\sim\pi_i^k}\Sp{f\in a_i}\geq\frac{\gamma}{2F}$ for all $f\in\bigcup_{a_i\in\A_i}a_i$. 

Properties of the IPS estimator are summarized in Lemma \ref{lmm:ips_estimator}. By using these properties, we can have the following lemma.
\begin{lemma}
\label{lmm:tilde_phi_concentration}
With probability $1-\delta$, for all $k\in[K]$, $i\in[m]$ and $a_i\in\A_i$, we have
$$\abs{\widetilde{\nabla}_i^k\Phi(a_i)-\nabla_i^k\Phi(a_i)}\leq \sqrt{\frac{4F^3\log(2mFK/\delta)}{\gamma \tau}}+\frac{2F^2\log(2mFK/\delta)}{\gamma \tau}.$$
\end{lemma}
\begin{proof}
By Lemma \ref{lmm:ips_estimator} and Bernstein's inequality, simultaneously for all $(i, k, f)\in[m]\times[K]\times\F$, with probability at least $1-\delta$, we have
$$\abs{[\wtthe^k_i(\pi^k)]_f-[\theta_i(\pi^k)]_f}\leq\sqrt{\frac{4F\log\Sp{2mFK/\delta}}{\gamma\tau}}+\frac{2F\log(2mFK/\delta)}{\gamma\tau}.$$
Since $\widetilde{\nabla}^k_i\Phi(a_i)=\sum_{f\in a_i}[\wtthe^k_i(\pi^k)]_f$, by triangle inequality, we have
$$\abs{\widetilde{\nabla}_i^k\Phi(a_i)-\nabla_i^k\Phi(a_i)}\leq \sqrt{\frac{4F^3\log(2mFK/\delta)}{\gamma \tau}}+\frac{2F^2\log(2mFK/\delta)}{\gamma \tau}.$$
\end{proof}

With this more refined gradient estimator, we can now have the following theorem.
\begin{theorem}
Let $T=K\tau$. For the congestion game with semi-bandit feedback, by running Algorithm \ref{algo:FW} with gradient estimator $\widetilde{\nabla}^k_i\Phi$ in \eqref{equ:grad_estimator_semi} and exploration distribution $\tilde{\rho}_i$ in Algorithm \ref{algo:compute_tilde_rho}, setting parameters $\nu=\frac{\sqrt{F}}{m\sqrt{K}}$, $\gamma=\frac{\sqrt{F}}{mK}$ and $\tau=K^2$, if $K\geq\frac{2\sqrt{F}}{m}$, then with probability $1-\delta$, we have
$$\text{Nash-Regret}(T)=\tau \sum_{k=1}^K G(\pi^k)=\tmco\Sp{m^2F^{3/2}T^{5/6}+m^3F^2T^{2/3}}.$$
\end{theorem}
\begin{proof}
By following the proof of Theorem \ref{theo:decentralized_bandit} and applying the concentration inequality in Lemma \ref{lmm:tilde_phi_concentration}, we can have
\begin{align*}
    \Phi(\pi^{k+1})\geq&\Phi(\pi^k)+(1-\gamma)\nu G(\pi^k)-\frac{(1-\gamma)\nu}{\sqrt{\gamma\tau}}\sqrt{4m^2F^3\log(2mK/\delta)}\\
    &\qquad-\frac{2(1-\gamma)\nu}{\gamma\tau}mF^2\log(mK/\delta)-\gamma 2m^2F-4m^3F(\nu^2+\gamma^2).
\end{align*}
Summing over $k\in[K]$ and we get
\begin{align*}
   \sum_{k=1}^K G(\pi^k)\leq&\frac{\Phi(\pi^{K+1})-\Phi(\pi^{1})}{(1-\gamma)\nu }+\frac{K}{\sqrt{\gamma\tau}}\sqrt{4m^2F^3\log(mK/\delta)}+\frac{2K}{\gamma\tau}mF^2\log(mK/\delta)\\
   &\qquad +\frac{2m^2FK\gamma}{(1-\gamma)\nu}+\frac{4(\nu^2+\gamma^2)m^3FK}{(1-\gamma)\nu}. 
\end{align*}
Set $\nu=\frac{\sqrt{F}}{m\sqrt{K}}$, $\gamma=\frac{\sqrt{F}}{mK}$, $\tau=K^2$ and notice that when $K\geq\frac{2\sqrt{F}}{m}$, we have $1-\gamma\geq\frac{1}{2}$. Thus, we can have
$$\sum_{k=1}^K G(\pi^k)=\tmco\Sp{m^2F^{3/2}K^{1/2}+m^3F^2}.$$
Then by Lemma \ref{lemma:gap to nash}, for $T=K\tau$, we have
$$\text{Nash-Regret}(T)=\tau \sum_{k=1}^K G(\pi^k)=\tmco\Sp{m^2F^{3/2}T^{5/6}+m^3F^2T^{2/3}}.$$
\end{proof}

\subsection{Lemmas for Semi-bandit Feedback}
\begin{algorithm}[ht]
\caption{Compute Exploration Distribution $\tilde{\rho}_i$}
\label{algo:compute_tilde_rho}
\begin{algorithmic}[1]
\STATE {\bf Input:} $\A_i$, player $i$-th action set
\STATE Initialize $\widetilde{\A}_i\leftarrow\emptyset$
\FOR{$a_i$ in $\A_i$}
    \IF{$\exists f\in a_i$ such that $f\notin \bigcup_{a_i'\in\widetilde{\A}_i}a_i'$}
        \STATE $\widetilde{\A}_i\leftarrow \widetilde{\A}_i\cup\Bp{a_i}$
    \ENDIF
    \IF{$\F_i=\bigcup_{a_i'\in\widetilde{\A}_i}a_i'$}
        \STATE {\bf break}
    \ENDIF
\ENDFOR
\STATE Assign $\tilde{\rho}_i(a_i)\leftarrow\frac{1}{2F}$ for each $a_i\in\widetilde{\A}_i$
\STATE Assign remaining probability mass arbitrarily to actions in $\A\setminus\widetilde{\A}_i$
\RETURN $\tilde{\rho}_i$
\end{algorithmic}
\end{algorithm}

\begin{lemma}
\label{lmm:tilde_rho}
Let $\F_i=\bigcup_{a_i\in\A_i}a_i$. For any player $i$, if $\tilde{\rho}_i$ is the output of Algorithm \ref{algo:compute_tilde_rho} and $\pi^k_i$ contains a mixture of $\tilde{\rho}_i$ with weight $\gamma$, then we have $\P_{a_i\sim\pi^k_i}\Sp{f\in a_i}\geq\frac{\gamma}{2F}$ for any $f\in\F_i$.
\end{lemma}
\begin{proof}
By Algorithm \ref{algo:compute_tilde_rho}, whenever a new action is added into $\widetilde{\A}_i$, it contains facility not appeared in current $\widetilde{\A}_i$. Then, since there are at most $\abs{\F_i}\leq F$ distinct facilities in the action set $\A_i$, the final $\widetilde{\A}_i$ must satisfy $|\widetilde{\A}_i|\leq F$. Therefore, $\tilde{\rho}_i$ is a valid distribution over $\A_i$.

Since $\pi^k_i$ contains a mixture of $\tilde{\rho}_i$ with weight $\gamma$, for any $a_i\in\A_i$, we have $\pi^k_i(a_i)\geq\gamma\tilde{\rho}_i(a_i)$. Thus, we have
\begin{align*}
    \P_{a_i\sim\pi^k_i}\Sp{f\in a_i}=&\sum_{a_i\in\A_i}\pi^k_i(a_i)\mathds{1}\Bp{f\in a_i}\\
    \geq & \gamma\sum_{a_i\in\A_i}\tilde{\rho}_i(a_i)\mathds{1}\Bp{f\in a_i}\\
    \geq & \gamma\sum_{a_i\in\widetilde{\A}_i}\tilde{\rho}_i(a_i)\mathds{1}\Bp{f\in a_i}\\
    =&\frac{\gamma}{2F}\sum_{a_i\in\widetilde{\A}_i}\mathds{1}\Bp{f\in a_i}\geq\frac{\gamma}{2F}.
\end{align*}
The last inequality above holds since by construction, $\widetilde{\A}_i$ contains all facilities contained in $\A_i$.

\end{proof}

\begin{lemma}
\label{lmm:ips_estimator}
If $\pi^k_i$ contains a mixture of $\tilde{\rho}_i$ given in Algorithm \ref{algo:compute_tilde_rho} with weight $\gamma$. Then, the IPS estimator $[\widetilde{\theta}^k_i(\pi^k)]_f$ satisfies
$$\E_k\Mp{[\wtthe^{k, t}_i(\pi^k)]_f}=[\theta_i(\pi^k)]_f,\quad |[\wtthe^{k, t}_i(\pi^k)]_f|\leq\frac{2F}{\gamma},\quad\text{and}\quad \E_k\Mp{[\wtthe^{k, t}_i(\pi^k)]_f^2}\leq\frac{2F}{\gamma}.$$
\end{lemma}
\begin{proof}
For the first property, since $\E_k\Mp{r^{k, t, f}\mid \a^{k, t}}=r^f(n^f(a_i^{k, t}, a_{-i}^{k, t}))$ and $\a^{k, t}\sim\pi^k$, We have
\begin{align*}
    &\E_k\Mp{[\wtthe^{k, t}_i(\pi^k)]_f}\\
    =& \E_{\a\sim\pi^k}\Mp{\frac{r^f(n^f(a_i, a_{-i}))\mathds{1}\Bp{f\in a_i}}{\P_{a_i'\sim\pi^k_i}(f\in a_i')}}\\
    =&\frac{1}{\P_{a_i'\sim\pi^k_i}(f\in a_i')}\cdot\E_{a_{-i}\sim\pi^k_{-i}}\Mp{\E_{a_i\sim\pi^k_i}\Mp{r^f(n^f(a_i, a_{-i}))\mathds{1}\Bp{f\in a_i}\mid a_{-i}}}\\
    =&\frac{1}{\P_{a_i'\sim\pi^k_i}(f\in a_i')}\cdot\E_{a_{-i}\sim\pi^k_{-i}}\Mp{\E_{a_i\sim\pi^k_i}\Mp{r^f(n^f(a_i, a_{-i}))\mid a_{-i}, f\in a_i}\P_{a_i\sim\pi^k_i}\Sp{f\in a_i\mid a_{-i}}}\\
    \overset{\text{(i)}}{=}&\frac{\P_{a_i\sim\pi^k_i}\Sp{f\in a_i}}{\P_{a_i'\sim\pi^k_i}(f\in a_i')}\cdot\E_{a_{-i}\sim\pikni}\Mp{r^f(n^f(a_{-i})+1)}\\
    =&[\theta_i(\pi^k)]_f.
\end{align*}
The equality (i) above holds because $\E_{a_i\sim\pi^k_i}\Mp{r^f(n^f(a_i, a_{-i}))\mid a_{-i}, f\in a_i}=r^f(n^f(a_{-i})+1)$ and $f\in a_i$ does not depend on $a_{-i}$.

For the second property, since $\P_{a_i\sim\pik_i}\Sp{f\in a_i}\geq\frac{\gamma}{2F}$ by Lemma \ref{lmm:tilde_rho} and $r^{k, t, f}\in[0, 1]$, we can immediately have $|[\wtthe^{k, t}_i(\pi^k)]_f|\leq\frac{2F}{\gamma}$.

For the third property, we have
\begin{align*}
    \E_k\Mp{[\wtthe^{k, t}_i(\pi^k)]_f^2}=&\frac{\E_{\a\sim\pi^k}\Mp{r^f(n^f(a_i, a_{-i}))^2\mathds{1}\Bp{f\in a_i}}}{\P_{a_i'\sim\pi^k_i}\Sp{f\in a_i'}^2}\\
    \leq & \frac{\E_{\a\sim\pi^k}\Mp{\mathds{1}\Bp{f\in a_i}}}{\P_{a_i'\sim\pi^k_i}\Sp{f\in a_i'}^2}\\
    = & \frac{\P_{a_i\sim\pi^k_i}\Sp{f\in a_i}}{\P_{a_i'\sim\pi^k_i}\Sp{f\in a_i'}^2}\\
    \leq & \frac{2F}{\gamma}.
\end{align*}
\end{proof}
\section{Algorithms for Independent Markov Congestion Games}
\label{sec:markov_algo_details}

In this section, present missing details of our centralized algorithm for independent Markov congestion games, which is summarized in Algorithm \ref{algo:nash_vi}. The proof of its theoretical guarantee is given in Appendix \ref{sec:ne_proof}.

\subsection{Algorithm for Semi-bandit Feedback}

Under the semi-bandit feedback, the players can receive reward information from all facilities they choose. Therefore, we can similarly define
\begin{align*}
    N^{k, f}_{h}(s^f, n)=&\sum_{k'=1}^{k}\mathds{1}\Bp{(s_{h}^{k', f}, n^f(\bm{a}^{k'}_h))=(s^f, n)},\\
    \hat{r}_{h}^{k, f}(s^f, n)=&\frac{\sum_{k'=1}^{k}r_{h}^{k', f}\mathds{1}\Bp{(s_{h}^{k', f}, n^f(\bm{a}^{k'}_h))=(s^f, n)}}{\Nhkf(s^f, n)\vee 1},\\
    \widehat{P}^{k, f}_h(s'^{f}\mid s^f, n)=&\frac{\sum_{k'=1}^{k}\mathds{1}\Bp{(s^{k', f}_{h+1}, s_{h}^{k', f}, n^f(\bm{a}^{k'}_h))=(s'^f, s^f, n)}}{\Nhkf(s^f, n)\vee 1}.
\end{align*}
Then, the estimators for the reward function and transition kernel can be defined as
\begin{equation}
    \label{equ:r_hat and P_hat}
    \hat{r}^{k}_{h, i}(s, \a)=\sum_{f\in a_i}\hat{r}^{k, f}_h(s^f, n^f(\a)),\quad \widehat{P}^k_h(s'\mid s, \bm{a})=\prod_{f\in\mc{F}}\widehat{P}^{k, f}_h(s'^f\mid s^f, n^f(\bm{a}))
\end{equation}


Then, with $\iota = 2\log(4(m+1)(\sum_{f\in\mc{F}}S^f)T/\delta)$, we define the bonus term to be $b_h^k(s, \bm{a})=b_h^{k, \mathrm{pv}}(s, \bm{a})+b_h^{k, \mathrm{r}}(s, \bm{a})$, which is a sum of transition bonus and reward bonus. In particular, we have
\fontsize{9.5}{9.5}
\begin{align}
     b_{h}^{k, \pv}(s, \a)=& \sum_{f\in\mc{F}}\sqrt{\frac{4H^2F^2S^f\iota}{\Nhkf(s^f, n^f(\a))\vee 1}}+\sum_{f\neq f'}\sqrt{\frac{4H^2F^2\Sp{S^fS^{f'}\iota}^2}{\Nhkf(s^f, n^f(\a))N_h^{k, f'}(s^{f'}, n^{f'}(\a))\vee 1}},\label{equ:pv_bonus}\\
	 b_{h}^{k, \reward}(s, \a)=&\sum_{f\in\mc{F}}\sqrt{\frac{\iota}{\Nhkf(s^f, n^f(\a))\vee 1}}.\label{equ:reward_bonus_1}
\end{align}
\normalsize

For convenience, we define $(\widehat{\P}^k_hV)(s, \a)=\E_{s'\sim \widehat{P}^k_h(\cdot\mid s, \a)}\Mp{V(s')}$ with value function $V:\S\mapsto\R$.

\begin{algorithm}[!t]
    \caption{Nash-VI for IMCGs}
    \label{algo:nash_vi}
    \begin{algorithmic}[1]
    \STATE {\bf Input:} $\epsilon$, accuracy parameter for Nash equilibrium computation
    \STATE {\bf Initialize:} $\overline{V}^k_{H+1, i}(s)=0$ for all $(i, k, s)\in[m]\times[K]\times\mc{S}$
    \FOR{episode $k=1, \dots, K$}
        \FOR{step $h=H, H-1, \dots, 1$}
            \FOR{player $i=1, \dots, m$}
                \STATE $\overline{Q}^k_{h, i}(s, \bm{a})\leftarrow\min\Bp{(\hat{r}^k_{h, i}+\widehat{\mathbb{P}}^k_h\overline{V}^k_{h+1, i}+b^k_h)(s, \bm{a}), HF}$ for all $(s, \bm{a})\in\mc{S}\times\mc{A}$\label{line:compute_nash}
            \ENDFOR
            \FOR{$s\in\mc{S}$}
                \STATE $\pi^k_h(\cdot\mid s)\leftarrow \epsilon\textsc{-Nash}(\overline{Q}^k_{h, 1}(s, \cdot), \cdots, \overline{Q}^k_{h, m}(s, \cdot))$
                \FOR{player $i=1, \dots, m$}
                    \STATE $\overline{V}^k_{h, i}(s)\leftarrow \E_{\bm{a}\sim\pi^k_h}[\overline{Q}^k_{h, i}(s, \bm{a})]$ 
                \ENDFOR
            \ENDFOR
        \ENDFOR
        \FOR{step $h=1, \dots, H$}
            \STATE Take action $\bm{a}^k_h\sim\pi^k_h(\cdot\mid s_h^k)$, observe reward $r^{k, f}_h$ and next state $s_{h+1}^k$
            \STATE Update reward estimator $\hat{r}^{k}_{h, i}$, transition estimator $\widehat{P}^{k}_h$ and bonus term $b_h^k$
        \ENDFOR
    \ENDFOR
    \end{algorithmic}
\end{algorithm}
\begin{remark}
Unlike Algorithm \ref{algo:nash_vi_matrix_game} for congestion game, here, $\oq^k_{h, 1}(s, \cdot), \dots, \oq^k_{h, m}(s, \cdot)$ in line \ref{line:compute_nash} of Algorithm \ref{algo:nash_vi} in general does not form a potential game. Therefore, we cannot use Algorithm \ref{algo:eps NE} and $\epsilon$-\textsc{Nash} is not always computationally efficient. 
\end{remark}

\subsection{Algorithm for Bandit Feedback}
In bandit feedback scenario, since players' observation about state transitions remains unaffected, we only need to modify the reward estimator $\hat{r}^k_{h, i}$ defined in \eqref{equ:r_hat and P_hat} and reward bonus term $b_{h}^{k, \reward}(s, \a)$ defined in \eqref{equ:reward_bonus_1}.

Similar to the congestion game with bandit feedback introduced in Section \ref{sec:centralized bandit}, for IMCGs, we can also write its reward function as $r_{h, i}(s, \a)=\inner{A_i(s, \a), \theta_h}$, where $\theta_h$ is unknown and $A_i(s, \a)$ is a 0-1 vector.

In particular, define $\theta_h\in[0, 1]^d$ with $d=m\sum_{f\in\mc{F}}S^f$ to be the vector such that $\theta_{h, i}=r_h^f(s^f, n)$ for some $f\in\mc{F}$ and $(s^f, n)\in\mc{S}^f\times[m]$. Then, we can similarly build estimator $\hat{r}^k_{h, i}$ through ridge regression as the following.\footnote{For the same reason, we take the regularization parameter in ridge regression to be 1.}

\begin{align}
    &\text{design matrix:}\quad V_h^k=I+\sum_{k'=1}^{k-1}\sum_{i=1}^{m}A_i(s_h^{k'}, \a_h^{k'})A_i(s_h^{k'}, \a_h^{k'})^\top,\label{equ:design_matrix}\\
    &\theta_h\text{ estimator:}\quad \widehat{\theta}_h^k=\Sp{V_h^k}^{-1}\sum_{k'=1}^{k-1}\sum_{i=1}^{m}A_i(s_h^{k'}, \a_h^{k'})r_{h, i}^{k'},\label{equ:theta_hat}\\
    &\text{reward estimator:}\quad \tilde{r}^k_{h, i}(s, \a)=\inner{A_i(s, \a), \widehat{\theta}^k_h},\label{equ:r_tilde}\\
    &\text{reward bonus:}\quad \tilde{b}_{h}^{k, \reward}(s, \a)=\max_{i\in[m]}\Norm{A_i(s, \a)}_{\Sp{V_h^k}^{-1}}\sqrt{\beta_k},\label{equ:reward_bonus_2}
\end{align}
where $\sqrt{\beta_k}=\sqrt{d}+\sqrt{Fd\log\Sp{1+\frac{mkF}{d}}+F\iota}$.
\section{Analysis for Algorithm \ref{algo:nash_vi}}
\label{sec:ne_proof}

\subsection{Bellman Equations for Genera-sum Markov Games}
Before analyzing Algorithm \ref{algo:nash_vi}, we first give a brief review of the Bellman equations for general-sum Markov games. These equations are well-known among the literature \cite{bai2020provable, liu2021sharp, jin2021vlearning}.

\paragraph{Fixed policies.} Given a fixed policy $\pi$, for any $(h, i, s, \bm{a})\in[H]\times[m]\times\mc{S}\times\mc{A}$, it holds that
\begin{equation}
	\label{equ:bellman_fixed}
	Q^{\pi}_{h, i}(s, \bm{a})=(r_{h, i}+\P_hV^\pi_{h+1, i})(s, \bm{a}),\quad V^{\pi}_{h, i}=\E_{\bm{a}'\sim\pi_h(\cdot\mid s)}\Mp{Q^\pi_{h, i}(s, \bm{a}')},
\end{equation}
where $V^{\pi}_{H+1, i}(s)=0$ for any $(i, s)\in[m]\times\mc{S}$.

\paragraph{Best responses.} Given a fixed policy $\pi$, define the best response value functions for player $i$ as $Q^{\dagger, \pi_{-i}}_{h, i}(s, \bm{a})=\max_{\pi_i\in\Delta(\mc{A}_i)}Q^{\pi_i, \pi_{-i}}_{h, i}(s, \bm{a})$ and $V^{\dagger, \pi_{-i}}_{h, i}(s)=\max_{\pi_i\in\Delta(\mc{A}_i)} V^{\pi_i, \pi_{-i}}_{h, i}(s)$. Then, for any $(h, i, s, \bm{a})\in[H]\times[m]\times\mc{S}\times\mc{A}$, it holds that
\begin{equation}
	\label{equ:bellman_best}
	\begin{split}
		&Q^{\dagger, \pi_{-i}}_{h, i}(s, \bm{a})=(r_{h, i}+\P_hV^{\dagger, \pi_{-i}}_{h+1, i})(s, \bm{a}),\\
		&V^{\dagger, \pi_{-i}}_{h, i}(s)=\max_{\nu\in\Delta(\mc{A}_i)}\E_{\bm{a}'\sim(\nu, \pi_{h, -i})(\cdot\mid s)}\Mp{Q^{\dagger, \pi_{-i}}_{h, i}(s, \bm{a}')},
	\end{split}
\end{equation}
where $V^{\dagger, \pi_{-i}}_{H+1, i}(s)=0$ for any $(i, s)\in[m]\times\mc{S}$.

\subsection{Proof of Theorem \ref{theo:ne}}
Recall that the update rule in Algorithm \ref{algo:nash_vi} is
$$\overline{Q}^k_{h, i}(s, \bm{a})\leftarrow\min\Bp{(\hat{r}^k_{h, i}+\widehat{\mathbb{P}}^k_h\overline{V}^k_{h+1, i}+b^k_h)(s, \bm{a}), HF},\quad \overline{V}^k_{h, i}(s)\leftarrow \E_{\bm{a}\sim\pi^k_h}[\overline{Q}^k_{h, i}(s, \bm{a})].$$
Similar to the proof of Theorem \ref{theo:ne_mg}, we define auxiliary value functions
\begin{equation}
    \label{equ:value_underline}
    \underline{Q}^k_{h, i}(s, \bm{a})\leftarrow\max\Bp{(\hat{r}^k_{h, i}+\widehat{\mathbb{P}}^k_h\underline{V}^k_{h+1, i}-b^k_h)(s, \bm{a}), 0},\quad \underline{V}^k_{h, i}(s)\leftarrow \E_{\bm{a}\sim\pi^k_h}[\underline{Q}^k_{h, i}(s, \bm{a})].
\end{equation}

We now begin to prove the first part of Theorem \ref{theo:ne}.

\begin{proof}[Proof of Theorem \ref{theo:ne}]
	\textbf{Step 1.} We first consider the setting of semi-bandit feedback. Assume the result in Lemma \ref{lmm:optimistic} holds since it is a high-probability event. Then, for any $(k, s)\in[K]\times\S$, it holds that
	$$\max_{i\in[m]}\Sp{V^{\dagger, \pikni}_{1, i}-V^{\pik}_{1, i}}(s)\leq \max_{i\in[m]}\Sp{\ov^k_{1, i}-\uv^k_{1, i}}(s) + H\epsilon.$$
	By the update rules in Algorithm \ref{algo:nash_vi}, we can notice the following recursive relations
	\begin{align*}
		&(\oq^k_{h, i}-\uq^k_{h, i})(s, \a)\leq\min\Bp{\hatphk(\ov^k_{h+1, i}-\uv^k_{h+1, i})(s, \a)+2b_h^k(s, \a), HF},\\
		&(\ov^k_{h, i}-\uv^k_{h, i})(s)=\E_{\a'\sim\pi^k_h(\cdot\mid s)}\Mp{(\oq^k_{h, i}-\uq^k_{h, i})(s, \a')}.
	\end{align*}
	Thus, we define $\tv^k_{H+1}(s)=0$ for any $s\in\S$ and $\tq^k_h$, $\tv^k_h$ recursively as
	\fontsize{9.5}{9.5}
	\begin{equation}
		\label{equ:tilde_qv}
		\tq^k_h(s, \a)=\min\Bp{(\hatphk\tv^k_{h+1})(s, \a)+2b_h^k(s, \a), HF},\quad \tv^k_h(s)=\E_{\a'\sim\pi^k_h(\cdot\mid s)}\Mp{\tq^k_h(s, \a')}.
	\end{equation}
	\normalsize
	Obviously, we have $\max_{i\in[m]}(\ov^k_{h, i}-\uv^k_{h, i})(s)\leq \tv^k_{H+1}$. Then, by inductively assuming the same relation holds for $h+1$, we can have
	\begin{align*}
		\max_{i\in[m]}(\oq^k_{h, i}-\uq^k_{h, i})(s, \a)=&\min\Bp{\max_{i\in[m]}\hatphk(\ov^k_{h+1, i}-\uv^k_{h+1, i})(s, \a)+2b_h^k(s, \a), HF}\\
		\leq & \min\Bp{(\hatphk\tv^k_{h+1})(s, \a)+2b_h^k(s, \a), HF}\\
		= & \tq^k_h(s, \a),\\
		\max_{i\in[m]}(\ov^k_{h, i}-\uv^k_{h, i})(s)\leq &\E_{\a'\sim\pi^k_h(\cdot\mid s)}\Mp{\max_{i\in[m]}(\oq^k_{h, i}-\uq^k_{h, i})(s, \a')}\\
		\leq & \E_{\a'\sim\pi^k_h(\cdot\mid s)}\Mp{\tq^k_h(s, \a')}\\
		= & \tv^k_h(s).
	\end{align*}
	Therefore, by induction, for any $h\in[H]$, we have
	$$\max_{i\in[m]}(\oq^k_{h, i}-\uq^k_{h, i})(s, \a)\leq\tq^k_h(s, \a), \quad\max_{i\in[m]}(\ov^k_{h, i}-\uv^k_{h, i})(s)\leq\tv^k_h(s).$$
	As a result, we have
	$$\text{Nash-Regret}(K)=\sum_{k=1}^{K}\max_{i\in[m]}\Sp{V^{\dagger, \pikni}_{1, i}-V^{\pik}_{1, i}}(s)\leq\sum_{k=1}^{K}\tv^k_1(s_1)+ HK\epsilon.$$
	
	\textbf{Step 2, Semi-bandit Feedback.} We define the martingale difference sequences
	\begin{align*}
		\M_h^k(\tq)&=\E_{\a'\sim\pi^k_h(\cdot\mid \shk)}\Mp{\tq^k_h(\shk, \a')}-\tq^k_h(\shk, \ahk),\\
		\M_h^k(\tv)&=(\P_h\tv^k_{h+1})(\shk, \ahk)-\tv^k_{h+1}(s_{h+1}^k).
	\end{align*}
	It is not hard to check that $\M_h^k(\tq)$ and $\M_h^k(\tv)$ are both indeed martingale difference sequences with respect to the history till episode $k$ and time step $h$.
	
	With these definitions, we can now decompose the regret bound as
	\begin{align*}
		\tv^k_h(\shk)=& \E_{\a'\sim\pi^k_h(\cdot\mid \shk)}\Mp{\tq^k_h(\shk, \a')}\tag{By \eqref{equ:tilde_qv}}\\
		=&\M_h^k(\tq)+\tq^k_h(\shk, \ahk)\\
		\leq &\M_h^k(\tq) +  2\bhk(\shk, \ahk) + (\hatphk\tv^k_{h+1})(\shk, \ahk)\tag{By \eqref{equ:tilde_qv}}\\
		\overset{\text{(i)}}{\leq} & \M_h^k(\tq) + 3\bhk(\shk, \ahk) + (\P_h\tv^k_{h+1})(\shk, \ahk)\\
		= & \M_h^k(\tq)+\M_h^k(\tv)+3\bhk(\shk, \ahk)+\tv^k_{h+1}(s^k_{h+1})
	\end{align*}
	The above inequality (i) holds by applying Lemma \ref{lmm:optimistic} and the fact $\tv_h^k(s)\leq HF$, which comes from the definition in \eqref{equ:tilde_qv}. Then, by unrolling this relation from $h=1$ to $h=H$ and noticing $\tv^k_{H+1}=\bm{0}$, we can have
	\begin{align}
		&\text{Nash-Regret}(K)\leq \sum_{k=1}^{K}\tv^k_1(s_1)+HK\epsilon\nonumber\\
		\leq & \sum_{k=1}^{K}\sum_{h=1}^{H}\Sp{\M_h^k(\tq)+\M_h^k(\tv)+3\bhk(\shk, \ahk)}+HK\epsilon\label{equ:before_final_regret}\\
		\leq & \tmco\Sp{HF\sqrt{T}}+3\sum_{k=1}^{K}\sum_{h=1}^{H}\bhk(\shk,\ahk)\tag{By Azuma-Hoeffding inequality and taking $\epsilon=1/T$.}\\
		\leq & \tmco\Sp{HF\sqrt{T}} +  6HF\sum_{f\in\mc{F}}\sum_{k=1}^{K}\sum_{h=1}^{H}\Sp{\sqrt{\frac{S^f\iota}{\Nhkf(\shkf, n^f(\ahk))\vee 1}}+\sqrt{\frac{\iota}{\Nhkf(\shkf, n^f(\ahk))\vee 1}}}\nonumber\\
		&\qquad +6HF\sum_{f\neq f'}S^fS^{f'}\sum_{k=1}^{K}\sum_{h=1}^{H}\sqrt{\frac{\iota^2}{\Sp{\Nhkf(\shkf, n^f(\ahk))N_h^{k, f'}(s_h^{k, f'}, n^{f'}(\a_h^{k, f'}))}\vee 1}}\nonumber\\
		\leq & \tmco\Sp{HF\sqrt{T}}+\tmco\Sp{\sum_{f\in\mc{F}}HFS^f\sqrt{mHT}}+\tmco\Sp{m^2H^2F\sum_{f\neq f'}\Sp{S^fS^{f'}}^2}\tag{By Lemma \ref{lmm:sum_sqrtn} and \ref{lmm:sum_sqrtnn}}\\
		\leq & \tmco\Sp{\sum_{f\in\mc{F}}FS^f\sqrt{mH^3T}}+\tmco\Sp{m^2H^2F\sum_{f\neq f'}\Sp{S^fS^{f'}}^2}.\nonumber
	\end{align}
	
	\textbf{Step 3, Bandit Feedback.} In the setting of bandit feedback, we only modify the reward estimator $\tilde{r}^{k}_{h, i}$ and its corresponding bonus term $\tilde{b}^{k, \reward}_h$. Thus, by going through the proof of Lemma \ref{lmm:optimistic}, we can notice that to have the same result for bandit feedback, it suffice to use Lemma \ref{lmm:least_square_bound} to show that the reward estimation error is bounded by the reward bonus term.
	
	Then, by the inequality \eqref{equ:before_final_regret}, we can notice that to achieve the final Nash-regret bound, we only need to bound the summation $\sum_{k=1}^{K}\sum_{h=1}^{H}\tilde{b}^{k, \reward}_h(\shk, \ahk)$, which is
	\begin{align*}
		\sum_{k=1}^{K}\sum_{h=1}^{H}\tilde{b}^{k, \reward}_h(\shk, \ahk)\leq & \sqrt{\beta_K}\sum_{k=1}^{K}\sum_{h=1}^{H}\max_{i\in[m]}\Norm{A_i(\shk, \ahk)}_{\Sp{V_h^k}^{-1}}\tag{By definition of $\tilde{b}^{k, \reward}_h$ in \eqref{equ:reward_bonus_2}.}\\
		\leq & \Sp{\sqrt{d}+\sqrt{Fd\log\Sp{1+\frac{mKF}{d}}+F\iota}}\tmco\Sp{H\sqrt{dFK}}\tag{By definition of $\beta_k$ and Lemma \ref{lmm:elliptical_potential}.}\\
		\leq & \tmco\Sp{d\sqrt{HF^2T}}\\
		= & \tmco\Sp{\sum_{f\in\mc{F}}mS^f\sqrt{HF^2T}}.\tag{Since $d=m\sum_{f\in\mc{F}}S^f$.}
	\end{align*}
	Therefore, by \eqref{equ:before_final_regret}, with $\epsilon=1/T$, under bandit feedback, we have
	\begin{align*}
		&\text{Nash-Regret}(K)\\
		\leq & \sum_{k=1}^{K}\sum_{h=1}^{H}\Sp{\M_h^k(\tq)+\M_h^k(\tv)+3\bhk(\shk, \ahk)}\\
		\leq & \tmco\Sp{\sum_{f\in\mc{F}}FS^f\sqrt{mH^3T}}+\tmco\Sp{m^2H^2F\sum_{f\neq f'}\Sp{S^fS^{f'}}^2}+\sum_{k=1}^{K}\sum_{h=1}^{H}\tilde{b}^{k, \reward}_h(\shk, \ahk)\\
		\leq & \tmco\Sp{\sum_{f\in\mc{F}}\Sp{\sqrt{mH^3}F+m\sqrt{HF^2}}S^f\sqrt{T}}+\tmco\Sp{m^2H^2F\sum_{f\neq f'}\Sp{S^fS^{f'}}^2}.
	\end{align*}
\end{proof}

\subsection{Lemmas for Semi-bandit Feedback}
The following two lemmas shows that our value function estimations are indeed optimistic.

\begin{lemma}
	\label{lmm:pv_bound}
	With probability at least $1-\delta$, simultaneously for arbitrary value function $V\in[0, HF]^{\mc{S}}$ and any tuple $(k, h, s, \a)$, it holds that $|(\widehat{\mathbb{P}}^k_h-\mathbb{P}_h)V(s, \bm{a})|\leq b_h^{k, \pv}(s, \a)$, where $b_h^{k, \pv}(s, \a)$ is defined in \eqref{equ:pv_bonus}.
\end{lemma}
\begin{proof}
	We define $\P_h^f$ to be the operator such that for some value function $V^f:\S^f\mapsto\R$, we have $(\P_h^fV^f)(s, \a)=\E_{s'^f\sim P_h^f(\cdot\mid s^f, n^f(\a))}\Mp{V^f(s'^f)}$. We also define $\widehat{\P}_h^{k, f}$ similarly. Then, by definition of our transition kernel, for operators $\P_h$ and $\widehat{\P}_h^k$, it holds that
	$$\P_h=\prod_{f\in\mc{F}}\P_h^f\quad\text{and}\quad\widehat{\P}_h^k=\prod_{f\in\mc{F}}\widehat{\P}_h^{k, f}.$$
	Therefore, by Lemma E.1 in \cite{chen2020efficient}, since $\Norm{V}_{\infty}\leq HF$, we have
	\begin{equation}
		\label{equ:lmm_pv_bound_1}
		\begin{split}
			|(\widehat{\mathbb{P}}^k_h-\mathbb{P}_h)V(s, \bm{a})|\leq&\sum_{f\in\mc{F}}\abs{(\widehat{\P}_{h}^{k, f}-\P_h^f)\Sp{\prod_{f'\neq f}\P_h^{f'}}V(s, \a)}\\
			&\qquad+2HF\sum_{f\neq f'}\mathrm{errp}^{k, f}_h(s, \a)\cdot\mathrm{errp}^{k, f'}_h(s, \a),
		\end{split}
	\end{equation}
	where $\mathrm{errp}^{k, f}_h(s, \a)=\|\widehat{P}^{k, f}_h(\cdot\mid s^f, n^f(\a))-P_h^f(\cdot\mid s^f, n^f(\a))\|_1$.
	
	Now, notice that $\Sp{\prod_{f'\neq f}\P_h^{f'}}V(s, \a)$ can be seen as some value function from $\S^f$ to $[0, HF]$. Therefore, by Lemma 12 in \cite{bai2020provable}, with probability at least $1-\frac{\delta}{2}$, simultaneously for any $V$ and $(k, h, s, \a)$, it holds that
	$$\abs{(\widehat{\P}_{h}^{k, f}-\P_h^f)\Sp{\prod_{f'\neq f}\P_h^{f'}}V(s, \a)}\leq 2HF\sqrt{\frac{S^f\iota}{\Nhkf(s^f, n^f(\a))\vee 1}},$$
	where $\iota = 2\log(4(m+1)(\sum_{f\in\mc{F}}S^f)T/\delta)$. Meanwhile, by standard Hoeffding's inequality and union bound, with probability at least $1-\frac{\delta}{2}$, simultaneously for any $(k, h, s, \a)$, it holds that
	$$\mathrm{errp}^{k, f}_h\leq S^f\sqrt{\frac{\iota}{\Nhkf(s^f, n^f(\a))\vee 1}}.$$
	Finally, by plugging above two concentration inequalities back into \eqref{equ:lmm_pv_bound_1}, we can have 
	$$|(\widehat{\mathbb{P}}^k_h-\mathbb{P}_h)V(s, \bm{a})|\leq b_h^{k, \pv}(s, \a).$$
\end{proof}

\begin{lemma}
	\label{lmm:optimistic}
	With probability at least $1-\delta$, for any $(k, h, i, s, \bm{a})\in[K]\times[H]\times[m]\times\S\times\A$, it holds that
	\begin{align}
		\oq^k_{h, i}(s, \bm{a})\geq Q^{\dagger, \pikni}_{h, i}(s, \bm{a})-(H-h)\epsilon,&\quad \uq^k_{h, i}(s, \bm{a})\leq Q^{\pik}_{h, i}(s, \bm{a}),\label{equ:q_optim}\\
		\ov^k_{h, i}(s)\geq V^{\dagger, \pikni}_{h, i}(s)-(H-h+1)\epsilon,&\quad \uv^k_{h, i}(s)\leq V^{\pik}_{h, i}(s.),\label{equ:v_optim}
	\end{align}
	where $\uq^k_{h, k}$ and $\uv^k_{h, i}$ are defined in \eqref{equ:value_underline}.
\end{lemma}
\begin{proof}
	The proof is adapted from \cite{liu2021sharp} and goes by induction from $h=H+1$ to $h=1$. We can see that inequalities \eqref{equ:v_optim} obviously hold when $h=H+1$ since by definition we have $\ov^k_{H+1, i}(s)=\uv^k_{H+1, i}(s)=0$ for any $(k, i, s)$. Now, suppose inequalities \eqref{equ:v_optim} hold for $h+1$. Then, if we have $\oq^k_{h, i}(s, \a)=HF$, it holds trivially that $\oq^k_{h, i}(s, \a)\geq Q^{\dagger, \pikni}_{h, i}(s, \a)$. Otherwise, by Bellman equations \eqref{equ:bellman_best} and update rule in Algorithm \ref{algo:nash_vi}, we have
	\begin{align*}
		&\oq^k_{h, i}(s, \a)-Q^{\dagger, \pikni}_{h, i}(s, \a)\\
		=& (\hat{r}^k_{h, i}-r_{h, i})(s, \a) + (\hatphk\ov^k_{h+1, i})(s, \a) - (\P_h V^{\dagger, \pikni}_{h+1, i})(s, \a)+b_h^k(s, \a)\\
		=&\underbrace{(\hat{r}^k_{h, i}-r_{h, i})(s, \a)}_{\text{(A)}} + \underbrace{\hatphk(\ov^k_{h+1, i}-V^{\dagger, \pikni}_{h+1, i})(s, \a)}_{\text{(B)}} + \underbrace{((\hatphk-\P_h)V^{\dagger, \pikni}_{h+1, i})(s, \a)}_{\text{(C)}}+b_h^k(s, \a).
	\end{align*}
	Now, recall that $b_h^k(s, \bm{a})=b_h^{k, \mathrm{pv}}(s, \bm{a})+b_h^{k, \mathrm{r}}(s, \bm{a})$. By reward definition in congestion game, we have
	$$(\hat{r}^k_{h, i}-r_{h, i})(s, \a)=\sum_{f\in a_i}(\hat{r}^{k, f}_{h, i}(s^f, n^f(\a))-r^f_{h, i}(s^f, n^f(\a))).$$
	Thus, by using standard Hoefding's inequality and union bound, we can immediately have $\abs{\text{(A)}}\leq b_h^{k, \reward}(s, \a)$. Then, since $V^{\dagger, \pikni}_{h, i}\in[0, HF]^{\mc{S}}$, by Lemma \ref{lmm:pv_bound}, we have $\abs{\text{(C)}}\leq b_h^{k, \pv}(s, \a)$. That is, we have $\text{(A)}+\text{(C)}+b_h^k(s, \a)\geq 0$.
	
	Then, by inductive hypothesis, we know that $\ov^k_{h+1, i}\geq V^{\dagger, \pikni}_{h+1, i}-(H-h)\epsilon$, which implies $\text{(B)}\geq 0$. Therefore, we have $\oq^k_{h, i}(s, \a)-Q^{\dagger, \pikni}_{h, i}(s, \a)\geq -(H-h)\epsilon$.
	
	For $\ov^k_{h, i}$ and $V^{\dagger, \pikni}_{h, i}$, we notice that in Algorithm \ref{algo:nash_vi}, $\pik$ is computed as the $\epsilon$-approximate Nash equilibrium of $(\oq^k_{h, 1}, \dots, \oq^k_{h, m})$. Therefore, it holds that 
	$$\ov^k_{h, i}(s)=\E_{\a\sim\pi^k_h(\cdot\mid s)}\Mp{\oq^k_{h, i}(s, \a)}\geq \max_{\nu\in\Delta(\mc{A}_i)}\E_{\bm{a}'\sim(\nu, \pi^k_{h, -i})(\cdot\mid s)}\Mp{\oq^k_{h, i}(s, \a')}-\epsilon.$$
	By Bellman equations \eqref{equ:bellman_best}, we also have
	$$V^{\dagger, \pikni}_{h, i}(s)=\max_{\nu\in\Delta(\mc{A}_i)}\E_{\bm{a}'\sim(\nu, \pi^k_{h, -i})(\cdot\mid s)}\Mp{Q^{\dagger, \pikni}_{h, i}(s, \a')}.$$
	Since $\oq^k_{h, i}(s, \a)-Q^{\dagger, \pikni}_{h, i}(s, \a)\geq -(H-h)\epsilon$, we immediately have $\ov^k_{h, i}(s)-V^{\dagger, \pikni}_{h, i}(s)\geq -(H-h+1)\epsilon$. Thus, by induction, we have that $	\oq^k_{h, i}(s, \bm{a})\geq Q^{\dagger, \pikni}_{h, i}(s, \bm{a})-(H-h)\epsilon$ and $\ov^k_{h, i}(s)\geq V^{\dagger, \pikni}_{h, i}(s)-(H-h+1)\epsilon$ for all $h\in[H]$.
	
	The inequalities for $\uv^k_{h, i}$ and $\uq^k_{h, i}$ can be proved similarly.
\end{proof}

\subsection{Additional Lemmas for Bandit Feedback}
The following lemma shows that the reward estimation error can be bounded by the reward bonus term.
\begin{lemma}
	\label{lmm:least_square_bound}
	With probability at least $1-\delta$, simultaneously for all $(i, k, h, s, \a)$, it holds that $|(\tilde{r}^k_{h, i}-r_{h, i})(s, \a)|\leq\tilde{b}^{k, \reward}_h(s, \a)$, where $\tilde{r}^k_{h, i}$ and $\tilde{b}^{k, \reward}_h$ are defined in \eqref{equ:r_tilde} and \eqref{equ:reward_bonus_2}.
\end{lemma}
\begin{proof}
	The proof is extremely similar to Lemma \ref{lmm:least_square_bound_mg}. By construction, we have
	\begin{align*}
		|(\tilde{r}^k_{h, i}-r_{h, i})(s, \a)|=&\abs{\inner{A_i(s, \a), \widehat{\theta}_h-\theta_h}}\\
		\leq & \Norm{A_i(s, \a)}_{\Sp{V_h^k}^{-1}}\Norm{\widehat{\theta}_h-\theta_h}_{V_h^k}\\
		\leq & \Norm{A_i(s, \a)}_{\Sp{V_h^k}^{-1}}\Sp{\Norm{\theta_h}_2+\sqrt{F\log\Sp{\det(V_h^k)}+F\iota}}.\tag{By Theorem 20.5 in \cite{lattimore2020bandit}.}
	\end{align*}
	Since each element in $\theta_h$ is bounded in $[0, 1]$ by construction, we have $\Norm{\theta_h}_2\leq\sqrt{d}$. 
	
	Then, by Lemma \ref{lmm:elliptical_potential_variant}, we have $\det\Sp{V_h^k}\leq \Sp{1+\frac{mkF}{d}}^d$ since by construction $\Norm{A_i(s, \a)}_2^2\leq F$.
	
	Finally, to make this bound valid for all player $i\in[m]$, we only need to take maximization over $i\in[m]$. Therefore, with probability at least $1-\delta$, we have
	$$|(\tilde{r}^k_{h, i}-r_{h, i})(s, \a)|\leq\max_{i\in[m]}\Norm{A_i(s, \a)}_{\Sp{V_h^k}^{-1}}\sqrt{\beta_k}=\tilde{b}^{k, \reward}_h(s, \a),$$
	where $\sqrt{\beta_k}=\sqrt{d}+\sqrt{Fd\log\Sp{1+\frac{mkF}{d}}+F\iota}$.
\end{proof}

The follow lemma bound the sum of reward bonus under bandit feedback.
\begin{lemma}
	\label{lmm:elliptical_potential}
	For any $h\in[H]$, it holds that
	$$\sum_{k=1}^{K}\max_{i\in[m]}\Norm{A_i(\shk, \ahk)}_{\Sp{V_h^k}^{-1}}\leq \tmco\Sp{\sqrt{dFK}},$$
	where $d=m\sum_{f\in\mc{F}}S^f$.
\end{lemma}
\begin{proof}
	First, since $V_h^k=I+\sum_{k'=1}^{k-1}\sum_{i=1}^{m}A_i(s_h^{k'}, \a_h^{k'})A_i(s_h^{k'}, \a_h^{k'})^\top$, we have $V_h^k\succeq I$ and thus $\Sp{V_h^k}^{-1}\preceq I$. Therefore, we have
	$$\Norm{A_i(\shk, \ahk)}_{\Sp{V_h^k}^{-1}}\leq\Norm{A_i(\shk, \ahk)}_I=\Norm{A_i(\shk, \ahk)}_2\leq \sqrt{F}.$$
	For simplicity, let $A_{h, i}^k=A_i(\shk, \ahk)$. Then, as a result, we have
	\begin{align*}
		\sum_{k=1}^{K}\max_{i\in[m]}\Norm{\Ahik}_{\Sp{V_h^k}^{-1}}=&\sum_{k=1}^{K}\min\Bp{\max_{i\in[m]}\Norm{\Ahik}_{\Sp{V_h^k}^{-1}}, \sqrt{F}}\\
		\leq & \sqrt{K\sum_{k=1}^{K}\min\Bp{\max_{i\in[m]}\Norm{\Ahik}_{\Sp{V_h^k}^{-1}}^2, F}}\\
		\leq &\sqrt{FK\sum_{k=1}^{K}\min\Bp{\max_{i\in[m]}\Norm{\Ahik}_{\Sp{V_h^k}^{-1}}^2, 1}}\\
        \leq & \sqrt{2FKd\log\Sp{1+\frac{mKF}{d}}}\tag{By Lemma \ref{lmm:elliptical_potential_variant}.}\\
        & = \tmco\Sp{\sqrt{dFK}}.
	\end{align*}
\end{proof}

\subsection{Technical Lemmas}

\begin{lemma}
	\label{lmm:sum_sqrtn}
	For any $f\in\mc{F}$, it holds that
	$$\sum_{k=1}^{K}\sum_{h=1}^{H}\sqrt{\frac{1}{\Nhkf(\shkf, n^f(\ahk))\vee 1}}\leq \tmco\Sp{\sqrt{mHS^fT}}.$$
\end{lemma}
\begin{proof}
	Here, we have
	\begin{align*}
		\sum_{k=1}^{K}\sum_{h=1}^{H}\sqrt{\frac{1}{\Nhkf(\shkf, n^f(\ahk))\vee 1}}=&\sum_{h=1}^{H}\sum_{s^f\in\mc{S}^f}\sum_{n=0}^{m}\sum_{\ell=1}^{N_h^{K, f}(s^f, n)}\sqrt{\frac{1}{\ell}}\\
		\leq & 2 \sum_{h=1}^{H}\sum_{s^f\in\mc{S}^f}\sum_{n=0}^{m}\sqrt{N_h^{K, f}(s^f, n)}\tag{By standard technique}\\
		\leq & 2\sqrt{(m+1)HS^f\sum_{h=1}^{H}\sum_{s^f\in\mc{S}^f}\sum_{n=0}^{m}N_h^{K, f}(s^f, n)}\\
		=&\tmco\Sp{\sqrt{mHS^fT}}.
	\end{align*}
	The last line above holds because $\sum_{h=1}^{H}\sum_{s^f\in\mc{S}^f}\sum_{n=0}^{m}N_h^{K, f}(s^f, n)=T$. This is based on a pigeon-hole principle argument. In particular, whenever the players take one more action, for any $f\in\mc{F}$, the count for some tuple $(h, s^f, n)$ will increase exactly by 1.
\end{proof}

\begin{lemma}[\cite{chen2020efficient}]
	\label{lmm:sum_sqrtnn}
	For any $f, f'\in\mc{F}$ and $f\neq f'$, it holds that
	$$\sum_{k=1}^{K}\sum_{h=1}^{H}\sqrt{\frac{1}{\Sp{\Nhkf(\shkf, n^f(\ahk))N_h^{k, f'}(s_h^{k, f'}, n^{f'}(\a_h^{k, f'}))}\vee 1}}\leq\tmco\Sp{m^2HS^fS^{f'}}.$$
\end{lemma}
\begin{proof}
	We define the joint empirical counter
	$$N_h^{k, f, f'}(s^f, s^{f'}, n, n')=\sum_{k'=1}^{k}\mathds{1}\Bp{(s_h^{k', f}, s_h^{k', f'}, n^f(\a_h^{k'}), n^{f'}(\a_n^{k'}))=(s^f, s^{f'}, n, n')}.$$
	Obviously, we have $N_h^{f, f'}(s^f, s^{f'}, n, n')\leq \min\Bp{\Nhkf(s^f, n), N_h^{k, f'}(s^{f'}, n')}$, which implies
	$$N_h^{k, f, f'}(s, s^{f'}, n, n')\leq\sqrt{\Nhkf(s^f, n)N_h^{k, f'}(s^{f'}, n')}.$$
	Therefore, we have
	\begin{align*}
		&\sum_{k=1}^{K}\sum_{h=1}^{H}\sqrt{\frac{1}{\Sp{\Nhkf(\shkf, n^f(\ahk))N_h^{k, f'}(s_h^{k, f'}, n^{f'}(\a_h^{k, f'}))}\vee 1}}\\
		\leq & \sum_{k=1}^{K}\sum_{h=1}^{H}\frac{1}{N_h^{k, f, f'}(\shkf, s_h^{k, f'}, n^f(\ahk), n^{f'}(\ahk))\vee 1}\\
		= & \sum_{h=1}^{H}\sum_{s^f\in\S^f}\sum_{s^{f'}\in\S^{f'}}\sum_{n=0}^{m}\sum_{n'=0}^{m}\sum_{\ell=1}^{N_h^{K, f, f'}(s^f, s^{f'}, n, n')}\frac{1}{\ell}\\
		= & \tmco\Sp{m^2HS^fS^{f'}}.
	\end{align*}
\end{proof}


\end{document}